%% file: Main.tex
\documentclass[conference]{IEEEtran}
 \IEEEoverridecommandlockouts
\usepackage{hyperref}
\usepackage{pifont}
\usepackage{balance} 
\usepackage{graphicx}

\usepackage{makecell}
\usepackage{multicol}
\usepackage{tablefootnote}
\usepackage{amsthm}

\usepackage{graphicx}

\input{commands}
 \newtheorem{theorem}{Theorem}
\usepackage{epstopdf}

\hypersetup{colorlinks=true,linkcolor=blue}
\newlength\myindent
\newcommand{\MYhref}[3][blue]{\href{#2}{\color{#1}{#3}}}%

 \hypersetup{breaklinks,colorlinks,citecolor=blue,linkcolor=blue}

 \usepackage{cleveref}
\begin{document}
%\title{No Strings Attached:  Efficient Certificate-Free Cryptosystems}
\title{Compatible Certificateless and Identity-Based Cryptosystems for Heterogeneous IoT}
%\titlerunning{Compatible Certificateless and Identity-Based Cryptosystems for  IoT}
\author{

	\IEEEauthorblockN{Rouzbeh Behnia}
	\IEEEauthorblockA{University of South Florida\\
		 Tampa, Florida, USA \\
		behnia@usf.edu}
	\and
	\IEEEauthorblockN{Attila A. Yavuz}
	\IEEEauthorblockA{University of South Florida\\
		Tampa, Florida, USA \\
		attilaayavuz@usf.edu}
		\and
\IEEEauthorblockN{Muslum Ozgur Ozmen$ ^\dagger $\thanks{$ \dagger $Work done in part when 	Muslum Ozgur Ozmen was at  University of South Florida.}} 
\IEEEauthorblockA{Purdue University\\ 
	West Lafayette, Indiana, USA \\
	mozmen@purdue.edu}

		\and

\IEEEauthorblockN{Tsz Hon Yuen} 
\IEEEauthorblockA{The University of Hong Kong\\ 
	Pokfulam, Hong Kong\\
	thyuen@cs.hku.hk}

}

\maketitle

\begin{abstract}
%Certificates are needed to ensure the authenticity of users' public keys.  However, the overhead introduced by   certificates (chains) might be too costly for some  IoT systems like   aerial drones. Certificate-free cryptosystems, like identity-based and certificateless systems, lift the burden of certificates  and   could be a suitable alternative for such IoT systems. However, despite their merits, there is a   research gap in proposing compatible identity-based and certificateless systems to  allow users from different domains (identity-based or certificateless) to communicate seamlessly.  Moreover,  more efficient constructions    would further  encourage  their   adoption in  resource-limited IoT systems.

Certificates ensure the authenticity of users' public keys, however their overhead  (e.g., certificate chains) might be too costly for some  IoT systems like   aerial drones. Certificate-free cryptosystems, like identity-based and certificateless systems, lift the burden of certificates  and   could be a suitable alternative for such IoTs. However, despite their merits, there is a   research gap in achieving compatible identity-based and certificateless systems to  allow users from different domains (identity-based or certificateless) to communicate seamlessly.  Moreover,  more efficient constructions  can enable their  adoption in  resource-limited IoTs.

In this work, we propose    new identity-based and certificateless cryptosystems that provide such  compatibility and efficiency. This  feature is beneficial for heterogeneous IoT settings  (e.g., commercial aerial drones), where different levels of trust/control is assumed on the trusted third party.    Our schemes are  more communication efficient than their public key  based  counterparts, as they do not need certificate processing. Our experimental analysis on both commodity and embedded IoT devices show that, only with the cost of having a larger system public key, our cryptosystems are more computation and communication efficient than their certificate-free counterparts. We prove the security of our schemes (in the random oracle model) and   open-source our cryptographic framework for public testing/adoption.

% \keywords{First keyword  \and Second keyword \and Another keyword.}
% 	\keywords{Identity-based cryptography \and  certificateless cryptography \and  IoT Systems \and lightweight cryptography} 
	
\end{abstract}


\input{introduction}

\input{buildingblock}

\input{Scheme}

\input{Sec-new}

%s\input{Sec-Kenny}

\input{perf}

%\input{tablesOzgur}

\input{RelatedWork}
%\input{Security_RejectionSampling}
%\input{Security_Asympototic} 
%\input{Security_params} 
%\input{Performance} 
%\input{Conclusion}

\bibliographystyle{IEEEtranS.bst}
 
{ \small
\bibliography{crypto-etc} } 
%\input{app}
%\begin{appendix}
%%	\renewcommand{\thesection}{\appendixname~\Alph{section}}
%	
%	\input{Appendix}
%\end{appendix}

%\input{extras}
\end{document}

%% file: commands.tex
\usepackage[para]{threeparttable}
\usepackage{graphicx}

\usepackage{textcomp}
\usepackage{nicefrac}
\usepackage[Symbol]{upgreek}

\usepackage{subcaption}
\usepackage{url}
\usepackage{graphicx}

\usepackage{epsfig}
\usepackage{url}
\usepackage{stfloats}
\usepackage{enumerate}
\usepackage{longtable}

\usepackage{latexsym}
\usepackage{verbatim}
\usepackage{amssymb}

\usepackage{xspace}

\usepackage{xcolor}

\usepackage{verbatim}
\usepackage{caption}
\usepackage{subcaption}
\usepackage{amsmath}
\usepackage{amssymb}
\usepackage{wrapfig}
\usepackage{tabularx}
\usepackage{lscape}
\usepackage{float}
\usepackage{inputenc}
\usepackage[utf8]{luainputenc}
\usepackage[T1]{fontenc}
\usepackage{pgfgantt}
\usepackage{footnote}
 \usepackage{algorithm}
\usepackage[noend]{algpseudocode}

\usepackage{multicol}
\usepackage{multirow}
\usepackage{tablefootnote}
\usepackage{mathtools}
\usepackage{graphicx}
\usepackage{caption}
\usepackage{subcaption}
\usepackage{tikz}
\usepackage{float}
\usepackage{bm}
\usepackage{url}
\usepackage{catoptions}
\newtheorem{Definition}{Definition}[section]
\newtheorem{Lemma}{{Lemma}}[section]

\usepackage{listings}% http://ctan.org/pkg/listings
\newcommand{\ignore}[1]{}
\usetikzlibrary{matrix,shapes,arrows,positioning,chains, calc}
\DeclarePairedDelimiter{\abs}{\lvert}{\rvert}

\mathchardef\mhyphen="2D

\usepackage{pbox}

\newcommand{\algrule}[1][.2pt]{\par\vskip.5\baselineskip\hrule height #1\par\vskip.5\baselineskip}

\newcommand{\AI}{\ensuremath {\mathcal{A}_I}{\xspace}}
\newcommand{\AII}{\ensuremath {\mathcal{A}_{II}}{\xspace}}
\newcommand{\A}{\ensuremath {\mathcal{A}}{\xspace}}

\newcommand{\Rq}{\ensuremath \stackrel{\$}{\leftarrow}\mathbb{Z}_{q}{\xspace}}

\newcommand{\Ra}{\ensuremath \stackrel{\$}{\leftarrow}{\xspace}}
\newcommand{\ZZ}{\ensuremath {\mathbb{Z}}{\xspace}}

\mathchardef\mhyphen="2D
\newcommand{\sk}{\ensuremath {\mathit{sk}}{\xspace}}

\newcommand{\pk}{\ensuremath {\mathit{PK}}{\xspace}}

\newcommand{\as}{\ensuremath {\leftarrow}{\xspace}}

\newcommand{\E}{\ensuremath {\mathcal{E}}{\xspace}}

\newcommand{\idsgnsetup}{\ensuremath {\texttt{IBS.Setup}}{\xspace}}
\newcommand{\idsgnextract}{\ensuremath {\texttt{IBS.Extract}}{\xspace}}
\newcommand{\sgnsig}{\ensuremath {\texttt{CLS.Sig}}{\xspace}}
\newcommand{\sgnver}{\ensuremath {\texttt{CLS.Ver}}{\xspace}}

\newcommand{\Adv}{\ensuremath { \mathcal{A}}}

{\xspace}

\newcommand{\eat}[1]{}                % Comment out large swaths of text
\newcommand{\m}{\ensuremath {m}{\xspace}}

\newcounter{linecounter}

\usepackage[hang,flushmargin]{footmisc}

\usepackage{pifont}% http://ctan.org/pkg/pifont
\newcommand{\Fq}{\ensuremath {\mathbb{F}_p}{\xspace}}
\newcommand{\EC}{\ensuremath {E(\Fq) }{\xspace}}
\newcommand{\Zp}{\ensuremath {\mathbb{Z}_q}{\xspace}}
\newcommand{\hash}{\ensuremath {\texttt{H}}{\xspace}}

\newcommand{\extract}{\ensuremath {\texttt{Extract}}{\xspace}}
\newcommand{\setup}{\ensuremath {\texttt{Setup}}{\xspace}}

\newcommand{\enc}{\ensuremath {\texttt{Enc}}{\xspace}}
\newcommand{\dec}{\ensuremath {\texttt{Dec}}{\xspace}}

\newcommand{\IDextract}{\ensuremath {\texttt{IBE.Extract}}{\xspace}}
\newcommand{\IDsetup}{\ensuremath {\texttt{IBE.Setup}}{\xspace}}

\newcommand{\IDenc}{\ensuremath {\texttt{IBE.Enc}}{\xspace}}
\newcommand{\IDdec}{\ensuremath {\texttt{IBE.Dec}}{\xspace}}

\newcommand{\ibs}{\ensuremath {\texttt{IBS}}{\xspace}}

\newcommand{\IBSExtract}{\ensuremath {\texttt{IBS.Extract}}{\xspace}}
\newcommand{\IBSSetup}{\ensuremath {\texttt{IBS.Setup}}{\xspace}}

\newcommand{\IBSSign}{\ensuremath {\texttt{IBS.Sign}}{\xspace}}
\newcommand{\IBSVerify}{\ensuremath {\texttt{IBS.Verify}}{\xspace}}

\newcommand{\Chall}{\ensuremath {\mathcal{C}}{\xspace}}

\newcommand{\cls}{\ensuremath {\texttt{CLS}}{\xspace}}
\newcommand{\cle}{\ensuremath {\texttt{CLE}}{\xspace}}

\newcommand{\CLEKGCsetup}{\ensuremath {\texttt{CLE.KGCSetup}}{\xspace}}
\newcommand{\CLEUserSetup}{\ensuremath {\texttt{CLE.UserSetup}}{\xspace}}
\newcommand{\CLEPartialKeyGen}{\ensuremath {\texttt{CLE.PartKeyGen}}{\xspace}}

\newcommand{\CLEUserKeyGen}{\ensuremath {\texttt{CLE.UserKeyGen}}{\xspace}}

\newcommand{\cleenc}{\ensuremath {\texttt{CLE.Enc}}{\xspace}}
\newcommand{\cledec}{\ensuremath {\texttt{CLE.Dec}}{\xspace}}

\newcommand{\CLSKGCsetup}{\ensuremath {\texttt{CLS.KGCSetup}}{\xspace}}
\newcommand{\CLSUserSetup}{\ensuremath {\texttt{CLS.UserSetup}}{\xspace}}
\newcommand{\CLSPartialKeyGen}{\ensuremath {\texttt{CLS.PartKeyGen}}{\xspace}}

\newcommand{\CLSUserKeyGen}{\ensuremath {\texttt{CLS.UserKeyGen}}{\xspace}}
\newcommand{\CLSSign}{\ensuremath {\texttt{CLS.Sign}}{\xspace}}
\newcommand{\CLSVerify}{\ensuremath {\texttt{CLS.Verify}}{\xspace}}
\newcommand{\IBESetup}{\ensuremath {\texttt{IBE.Setup}}{\xspace}}
\newcommand{\IBEExtract}{\ensuremath {\texttt{IBE.Extract}}{\xspace}}

\newcommand{\IBEEnc}{\ensuremath {\texttt{IBE.Enc}}{\xspace}}
 \newcommand{\IBEDec}{\ensuremath {\texttt{IBE.Dec}}{\xspace}}

\newcommand{\ListHI}{\ensuremath {\texttt{List}_{\texttt{H}_1}{\xspace}}}
 
 \newcommand{\ListHII}{\ensuremath {\texttt{List}_{\texttt{H}_2}{\xspace}}}
 
 \newcommand{\ListHIII}{\ensuremath {\texttt{List}_{\texttt{H}_3}{\xspace}}}
 
 \newcommand{\ListHIV}{\ensuremath {\texttt{List}_{\texttt{H}_4}{\xspace}}}

%% file: introduction.tex
\vspace{-3mm}
\section{Introduction} 
%Public key cryptosystems form the foundation of scalable key management, and therefore constitute the backbone of secure communication in the Internet as well as the emerging Internet of Things (IoT) applications. 
%Public Key Infrastructure  (PKI) is necessary to ensure  the authentication of public keys to prevent man-in-the-middle attacks. This is done via certificates that are issued on user-generated public keys by the  certification authorities.
% Due to its simplicity and wide-spread use, PKI has been serving as a backbone for secure communication in the standard Internet settings. 
%However, despite its advantages and broad use,  the adoption of  PKI in some  emerging and essential  IoT applications might be highly cumbersome and costly.

Mobile and heterogeneous IoT applications harbor large quantities of resource-limited and non-stationary IoT devices, each with different  capabilities, configurations, and user domains. For instance, emerging commercial aerial drone network  protocols\footnote{\url{https://github.com/mavlink/mavlink}}   need a near real-time communication and processing over a    bandwidth-limited network. 
%Similar requirements might be observed in other heterogeneous IoT applications such as vehicular networks and energy delivery systems. 
 There are multiple hurdles of relying on   traditional PKI for  such systems: (i) The maintenance of   PKI for  such IoT networks demands a substantial infrastructure investment  \cite{ShamirID-Based}. (ii)     PKI requires transmission and verification of  certificate chains at the sender's/verifier's side.  This communication and computation overhead  could create a major bottleneck for mobile IoT devices (e.g., aerial drones \cite{DBLP:journals/corr/abs-1904-06829}) that potentially need to interact with a   number of devices. In certain cases, these certificate chains might be larger than the actual measurements/commands being transmitted and therefore, might be the dominating cost for these  applications.  \autoref{fig:pki}-a depicts a high-level illustration of traditional PKI for mobile IoT applications.

 Identity-based (IDB)  and certificateless (CL) cryptosystems offer implicit certification~\cite{PatersonCertLess,ShamirID-Based,BonehID-based}, and therefore can mitigate the aforementioned hurdles. In IDB, the user's public key  is derived from their  identifying information, and the system relies on a fully-trusted third party (TTP), called the private key generator (PKG), to issue users' private keys.  The top portion of \autoref{fig:pki}-b depicts IDB encryption, wherein the user authenticates itself to the PKG and receives a private key  corresponding to its identity $ D_1 $. The sender can use $ D_1 $ as the  public key to run encryption. IDB is potentially suitable for applications  where the system setup is done and managed by a trusted centralized entity. In CL systems~ \cite{PatersonCertLess} the trust on the TTP is   lowered by allowing the private key of the user to consist of two parts. One is computed by the user and the other is by the TTP  (called  the KGC). The bottom portion of  \autoref{fig:pki}-b outlines CL encryption, where  the user computes its  key pair and then works as in IDB to receive the other part of the private key from the KGC.  CL cryptosystems are suitable for architectures that might not assume a fully trusted third party where  the trust level on the KGC is similar to traditional certification authorities.
% They can reach to the same trust-level as in traditional certification authorities.  

 IDB and CL cryptosystems have their own merits and drawbacks, and therefore might be used in different IoT applications. Hence, it is expected that there will be different user groups who rely on IDB and CL  cryptosystems initiated in different domains/systems. For example,  Amazon's Prime Air\footnote{\url{https://www.amazon.com/Amazon-Prime-Air/b?ie=UTF8&node=8037720011}} would require drones,  under the complete control of Amazon, to interact with other drones (e.g., personal) to ensure safe operation. By employing IDB cryptography on its drones, Amazon can have complete control over the operations of its delivery drones while avoiding the overheads of traditional PKI. However, it is a strong assumption that other drones, outside Amazon's network, will adopt a similar cryptographic setting to ensure safe and secure operations. For instance, personal users rarely trust any third party to have complete control and knowledge of their drones' activity.  To the best of our knowledge,  there is a  significant research gap in enabling  a seamless communication between users who are registered under different domains (e.g.,  IDB and CL). This is a potential obstacle to  widely deploy efficient certificate-free solutions in heterogeneous environments. This limitation is mentioned in \autoref{fig:pki}-b. Moreover, it is important to further improve the computational efficiency of IDB and CL techniques to offer a low end-to-end delay that is needed by delay-aware IoT applications.

 \noindent\textbf{Our Contribution.} We propose a new series of  public key encryption, digital signature, and key exchange schemes that permit users from different domains (IDB or CL)  to communicate seamlessly.  To our knowledge,  this is the first set of certificate-free cryptosystems that achieve such    compatibility and   efficiency, and therefore a suitable alternative for resource-limited IoT systems such as commercial aerial drones. The   idea behind our    constructions is to create special key generation algorithms that harness the additive homomorphic property of the exponents  and cover-free functions to enable the users to  incorporate their private keys into the one provided by the  TTP  without falsifying it. As detailed in \autoref{sec:ProposedSchemes},  this special design is applicable across   our IDB and CL algorithms,  and therefore it permits a seamless communication between our IDB and CL cryptosystems. This strategy   also   reduces the cost of online operations and enables our schemes to achieve a lower end-to-end delay compared to their counterparts.  We elaborate on some desirable properties of our schemes as below.

 \begin{figure*} [t]
 
 	 		\caption{Proposed IDB and CL Cryptosystems and Alternatives (High-Level)}\label{fig:pki} 	
 
 	\centering
% 	\begin{minipage}[b]{.45\textwidth}
 		\includegraphics[width=0.85\linewidth]{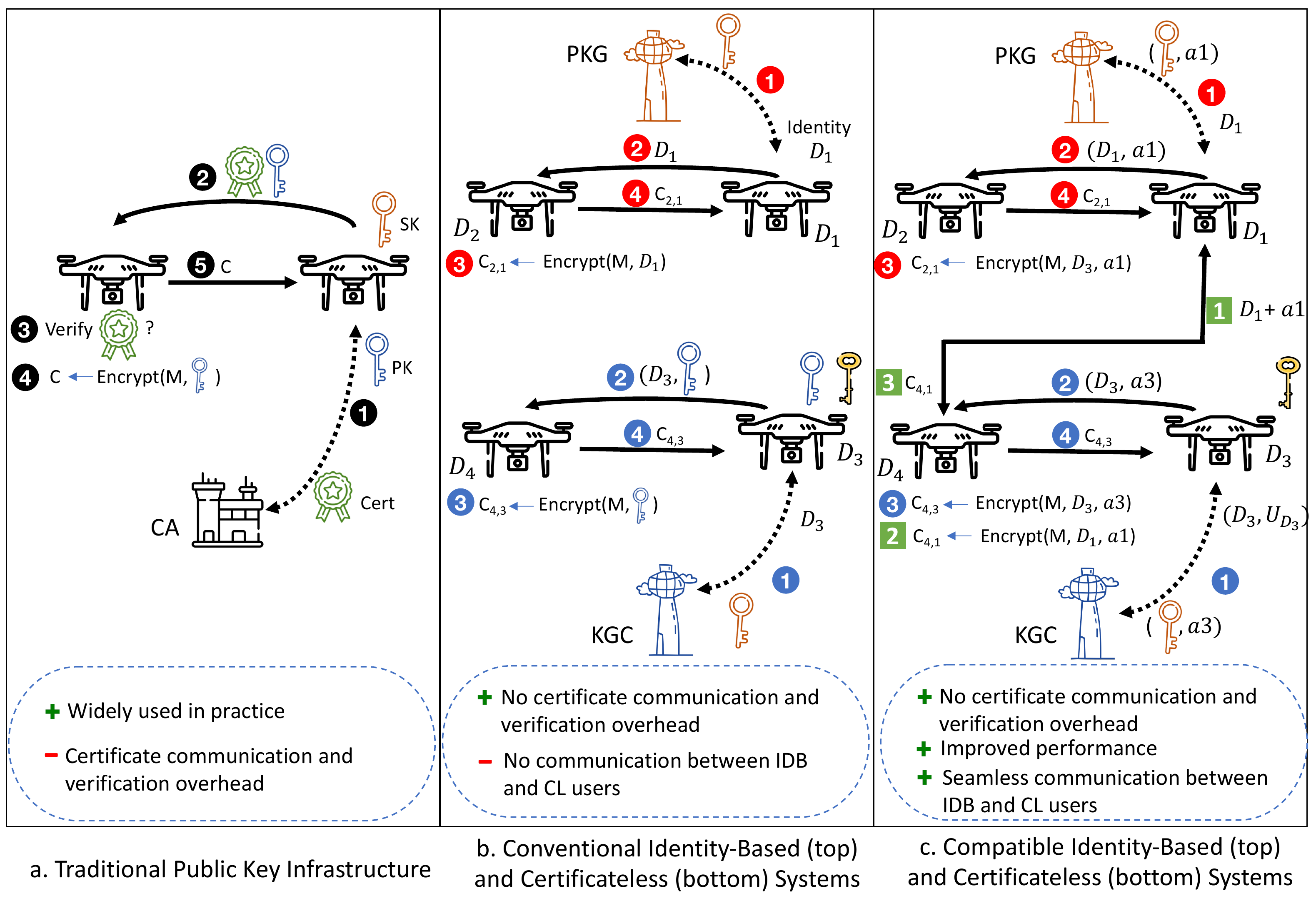} 

 	 \end{figure*}
 
\noindent  $\bullet \;$ \underline{Compatible IDB  and CL Schemes:} Fig  \ref{fig:pki}.c outlines the  concept of \emph {compatible}  IDB and CL schemes where the users from different domains (and trust-levels) can use identical encryption, signature,  and key exchange algorithms to communicate without any  additional overhead.

   \noindent$\bullet \;$ \underline{Computation \& Communication Efficiency:} Based on our analysis, new schemes offer performance advantages over their counterparts: (i) Similar to other IDB/CL cryptosystems, our schemes lift the hurdle of certificate transmission and verification, and therefore offer significant communication efficiency  over some of the most efficient PKI-based schemes. This advantage grows proportional to the size of the certificate chain. (ii) Our  schemes   outperform their certificate-free counterparts on the vast majority of the performance metrics. For instance, the end-to-end delay in our IDB/CL encryption schemes is $ \approx25$\%  lower than our most efficient counterpart in \cite{BertinoCertLessDrone}. Our signature schemes achieve up to $ \%52 $ faster end-to-end delay as compared to our counterparts. We also achieve a $ 65\% $ lower end-to-end delay for our key exchange schemes. 
   
%   Similar to other IDB and CL systems, our IDB and CL systems do not incur the communication and computation overheads to transmit and verify  certificates   as in traditional public key cryptosystems.  Moreover, we achieve a better efficiency than our most efficient counterparts in the IDB and CL setting. More specifically, the end-to-end delay in our IDB/CL encryption schemes is $ \approx25$\%  lower than our most efficient counterpart in \cite{BertinoCertLessDrone}. Our signature schemes achieve up to $ \%52 $ faster end-to-end delay as compared to our counterparts. We also achieve $ 65\% $ lower end-to-end delay for our key exchange scheme. 
    
% 	   $\bullet \;$  \underline{Provable Security:} The proposed schemes are provably secure under the hardness of standard discrete logarithmic based problems. Our proofs are in the random oracle model. For our signature schemes we utilize the concept of forking lemma \cite{Pointcheval1996}  in our proofs.
 	
 	 \noindent  $\bullet \;$  \underline{Open-Sourced   Implementation:} We   implemented our schemes on a commodity hardware and    an  8-bit AVR microprocessor, and  compared their performance with a   variety of their counterparts capturing some of the most efficient traditional PKI, IDB and CL schemes (see \autoref{sec:Performance} for details).    We  open-source our implementations for broad testing, benchmarking, and adoption purposes.

%% file: buildingblock.tex
\section{Preliminaries}\label{Prelim}
\noindent \textbf{Notation.}  Given two primes $ p $ and $ q $, we define a finite field \Fq~and a group \Zp.  We work on \EC~as an elliptic curve (EC) over   \Fq, where  $ P\in \EC $ is the generator of the points on the curve.  We denote a scalar and a point on a curve with small  and capital letters, respectively.   $ x\Ra S $ denotes a random uniform selection of  $ x  $ from a set $ S $.  We define the bit-length of a variable  as $ |x| $ (i.e., $ |x| = \log_2 x $). EC scalar multiplication is denoted as $ xP $, and all EC operations use an additive notation.  Hash functions  are $ \hash_1$: $ \EC \times \EC \rightarrow \{0,1\}^\gamma$,  $ \hash_2$: $ \{0,1\}^n \times \{0,1\}^*  \rightarrow \ZZ_q, \hash_3$: $\EC \rightarrow \{0,1\}^n$  $  \hash_4$: $ \{0,1\}^n  \rightarrow \{0,1\}^n$ and $  \hash_5$: $ \{0,1\}^n   \times \EC \rightarrow \ZZ_q$, where all hash functions are random oracles~\cite{RandomOracleModel93}.  FourQ \cite{FourQ} is a special EC that is defined by the complete twisted Edwards equation~$ \E/F_{p^2} : -x^2+y^2 = 1 +dx^2y^2 $. FourQ is known to be one of the fastest elliptic curves that admits $128$-bit security level~\cite{FourQ}. Moreover, with extended twisted Edwards coordinates, FourQ offers the fastest EC addition algorithms~\cite{FourQ}, that is extensively used in our optimizations. All of our schemes are realized on FourQ. \\
\vspace{1mm}
\noindent \textbf{Definitions.}\label{subsec:Certless}\label{subsec:ID-based}
We first give our intractability assumptions followed by the definitions of identity-based and certificateless encryption and signature schemes.  
%In identity-based systems, the user's public key is derived from a public identity $ ID $  and their private key is computed by the PKG  (assumed to be fully trusted) . %by calling the $ \texttt{Extract}(\cdot) $  algorithm.  We formally define the notion of  identity-based encryption in the following definition. 

%In this section, we first define some intractable problems that we will use in our security proof. We then define the notions of identity-based and certificateless encryption and signature.  
%In identity-based systems, the public key of the user is derived from the publicly identifying information $ ID $  and their private key is computed by the PKG  (assumed to be fully trusted)  by calling the $ \texttt{Extract}(\cdot) $  algorithm.  We formally define the notion of  identity-based encryption in the following definition. 

\begin{Definition} \label{def:ECDLP}
	Given points $ P,Q \in \EC $,  the Elliptic Curve Discrete Logarithm Problem (ECDLP) asks to find $ a $,  if  it exists, such that $ aP (\bmod \,p)=Q $. 	
	\end{Definition}
\begin{Definition} \label{def: CDH}
	Given   $ P,aP,bP \in \EC $,  the   Computational Diffie-Hellman Problem (CDHP) asks to compute $ abP $. 
\end{Definition}

\begin{Definition} \label{def:ID-baasedEnc}
	An identity-based encryption scheme is  consisted of four   algorithms $\textnormal{\texttt{IBE} = \{\setup, \extract,  \enc, \dec \}}$.

	\noindent\textnormal{$\underline{	(msk, params)\gets \IDsetup(1^\kappa) }$}: Given the security parameter $ \kappa $, the PKG   selects master secret key $ msk $,    computes master public key $ mpk$ and system parameters $ params $ (an implicit  input to all the following algorithms).
	
	\noindent  \textnormal{$\underline{(sk_{ID} ,Q_{ID} )\gets \IDextract(ID,msk)} $:}  Given an identity $ ID $ and $ msk $, the PKG computes  the commitment value  $ Q_{ID}  $ and the private key $ sk_{ID} $. 
	
	\noindent  \textnormal{$\underline{c \gets \IDenc(m,ID,Q_{ID} ) }$: }Given  a message $ m $ and  $ (ID,Q_{ID}) $,   the sender computes the ciphertext $ c $.

	\noindent  \textnormal{$\underline{ m \gets \IDdec(sk_{ID},c) }$:} 	Given  the ciphertext $ c$  and the private key of the   receiver $ sk_{ID} $,   the receiver  returns either the corresponding plaintext $ m$ or $\perp$ (invalid).
\end{Definition}

\begin{Definition} \label{Def:IBS}
	An identity-based signature scheme~is  defined by four algorithms 	\textnormal{$\ibs = \{\texttt{Setup},\texttt{Extract},\texttt{Sig},\texttt{Ver}\}$}.
	
	\itemsep0.4em

	\noindent  \textnormal{$\underline{(msk, mpk, params)\leftarrow \idsgnsetup(1^{\kappa})}$:}  As in \textnormal{\IDsetup}~in   \autoref{def:ID-baasedEnc}.
	
	\noindent  \textnormal{$\underline{(sk_{ID},Q_{ID} )\gets \idsgnextract(ID,msk) }$:}  As in \textnormal{\IDextract}~in   \autoref{def:ID-baasedEnc}.

	\noindent  \textnormal{$\underline{  {\sigma\leftarrow \IBSSign(\m,\sk_{ID})}}$:} Given a message $\m $~and $ \sk_{ID} $,  returns a
	signature $\sigma$.
	
	\noindent  \textnormal{$\underline{{ d\leftarrow \IBSVerify(\m,ID,Q_{ID} ,\sigma)}}$:} Given $\m $,   $\sigma$ and  $ (ID,Q_{ID}) $ as input, if the signature is valid, it returns $d=1$, else $d=0$.% it returns a bit $ d\in \{0,1\}$.
	
\end{Definition}

%In the following, we define the notions of certificateless  encryption and signature schemes.%& along with their  security models.  
%Certificateless cryptography was proposed \cite{PatersonCertLess} to relax the trust requirement on the TTP  in  identity-based systems. This is achieved by computing the user's private key from two independent components:   one provided by the KGC  called the \emph{partial private key} and the other computed by the user called the \emph{secret value}.  

\begin{Definition} \label{def:CertlessEnc}
	A certificateless  encryption scheme is defined by six  algorithms \textnormal{$ {\cle} = \{\texttt{KGCSetup},\texttt{UserSetup},  \texttt{PartKeyGen}, \linebreak \texttt{UserKeyGen},\enc, \dec \} $}.

	\noindent  \textnormal{$\underline{	(msk, mpk, params)\gets \CLEKGCsetup(1^\kappa) }$:} Give the security parameter $ \kappa $, the KGC generates master secret key $ msk$,  master public key $ mpk $ and the system parameters $ params$ (an implicit  input to all the following algorithms).
	 
	\noindent  \textnormal{$\underline{(\alpha, U) \gets \CLEUserSetup(\cdot)} $:}  The user    $ ID $ computes her secret value $ \alpha $ and its corresponding commitment $ U.$

\noindent  \textnormal{$\underline{(w,Q_{ID}) \gets \CLEPartialKeyGen(ID,U,msk)} $:}  Given $ ID $,  $ U $,   and   $ msk $ , the KGC computes partial private key  $ w $ and its corresponding public commitment $ Q_{ID} $.
	
	\noindent  \textnormal{$\underline{x_{ID} \gets \CLEUserKeyGen(w,\alpha)} $: } Given  $(w, \alpha) $, the user $ ID $ computes   $ 	x_{ID} $. 
	
	\noindent  \textnormal{$\underline{c \gets \cleenc(m,ID, Q_{ID}) }$: }Given   $ (m,ID,Q_{ID})$,  sender computes ciphertext $ c $. 
	
	\noindent  \textnormal{$\underline{ m' \gets \cledec(x_{ID},c) }$:} Given  the ciphertext $ c$  and the private key of the   receiver  $ x_{ID} $,     the receiver  returns either corresponding plaintext $ m$ or $\perp$ (invalid).
\end{Definition}

\begin{Definition} \label{def:CertlessSig}
	A certificateless signature scheme~is  defined by six algorithms 	\normalfont{$\texttt{CLS} =  \{\texttt{KGCSetup},\texttt{UserSetup},  \texttt{PartKeyGen}, \texttt{UserKeyGen},\linebreak \texttt{Sig}, \texttt{Ver} \} $}. The definition of algorithms are as in   \autoref{def:CertlessEnc} except for $ \texttt{(CLS.Sig,CLS.Ver)}$.

	\noindent  $\underline{  {\sigma\leftarrow \sgnsig(\m,x_{ID})}}$: Given  a message $\m $, and the signer's private key $ x_{ID} $, it  returns a	signature $\sigma$.
	
	\noindent  $\underline{{  d\leftarrow \sgnver(\m,ID,Q_{ID},\sigma)}}$: Given  $ m $,   $\sigma$ and   $ (ID,Q_{ID}) $ as input, if the signature is valid, it returns $d=1$, else $d=0$.% it returns a bit $ d\in \{0,1\}$.
	
\end{Definition}
 \newcommand{\ExptCLE}{\ensuremath { \mathit{Expt}^{\textit{IND\mbox{-}CLE\mbox{-}CCA}}_{\A}}{\xspace}}

\newcommand{\ExptCLS}{\ensuremath { \mathit{Expt}^{\textit{EU\mbox{-}CLS\mbox{-}CMA}}_{\A}}{\xspace}}
   \newcommand{\oracle}{\ensuremath {\mathcal{O}}{\xspace}}
\section{Security Model}
 The security model of identity-based schemes is slightly stronger than   those for traditional PKI based schemes. More specifically,  the adversary can query for the private key of any user $ ID $, except for the target user $ ID^* $. In this paper, we constructed our schemes by following the security model of Identity-based systems proposed in \cite{BonehID-based}. In certificateless systems,  the private key of the users consists of two parts: (i) user secret key $ \alpha $, which is selected by the user, and (ii) partial private key $ w $, which is supplied to the user by the KGC.  Therefore, following \cite{PatersonCertLess}, it is natural to consider two types of adversaries for such systems. 

A Type-I adversary    $ \A_I  $ does not have access to   $ msk $  or the user's partial private key $ w $ but is able to replace any user's public key $ U $ with  public key of its choice $ U'$. However, in our security model, since we adopt the binding method \cite{PatersonCertLess}, replacing the public key will result in falsifying the partial private key (and evidently the private key). Therefore, following \cite{ManHo07}, we allow  $ \A_I  $  to query for the secret key of the user via $ \alpha \gets \oracle_{\texttt{SecKey}}(ID) $. Note that our model can also be extended to allow  $ \A_I  $ to replace the public key of the user (see \autoref{Sec:SecAnalysis}). A Type-II adversary  $ \A_{II}  $ is assumed to be a  malicious KGC.  Having knowledge on $msk $,  $ \A_{II}  $ can query  the partial private key of the user  via $ w \gets \oracle_{\texttt{PartKey}}(ID) $. Following \cite{PatersonCertLess}, we allow the adversary $ \A\in \{\A_I,\A_{II}\} $  to  extract private key of users' private keys via the $x_{ID}\gets \oracle_{\texttt{Corrupt}}(ID) $.    We note that inspired by \cite{BaekCertLessPairingFree}, many improvements on the security models of certificateless systems have been suggested  (e.g., \cite{ManHo07,BaekCertLessPairingFree,CertLessRevisited}). In this paper, we provide our proof in the original model proposed in  \cite{PatersonCertLess,BaekCertLessPairingFree}, but note that many of those stronger security requirements can be enforced if needed. We assume $ params $ to be the implicit input to all the following algorithms. 
%To consider both types of adversary, we consider $ \mathcal{O} (\cdot)  $     as  either a secret key query or partial private key query when   considering the experiment against $ \A_I $ or $ \A_{II} $, respectively. 

%\{\KGCsetup,  \UserSetup,   \PartialKeyGen,  \linebreak \UserKeyGen,\enc, \dec \} 
%Given the above oracles, in the following we define the security of a certificateless encryption scheme and signature scheme. 
\begin{Definition}\label{def:CertlessINDCCA} 
	The indistinguishability of a $  {\cle}$   under chosen ciphertext attack	(IND-CLE-CCA) experiment \ExptCLE  is defined  as follows.
	
%		The indistinguishability of a CLE under chosen ciphertext attack	(IND-CLE-CCA) experiment \ExptCLE~for  an certificateless encryption  scheme   $ {\cle}  $   is defined in the following.
%	
	\begin{itemize}\normalfont
		\item [--] \Chall~runs \CLEKGCsetup($ 1^\kappa $) and returns $ mpk  $ to \A. 
		\item[--] $ (ID^*,m_0, m_1)\leftarrow \Adv^{\oracle_{\texttt{PartKey}} ,\oracle_{\texttt{SecKey}},\oracle_{\texttt{Corrupt}},\oracle_{\texttt{Dec}}} (mpk) $
		\item [--] \Chall~picks $ b\Ra\{0,1\} $,   $ c_b \gets \cleenc(m_b,ID^*)$ and returns $ c_b  $ to $ \mathcal{A}. $
		\item [--] \A~performs the second series of queries, with a restriction of querying $ ID^*  $ or $ c_b $ to  $  \texttt{Corrupt}(\cdot)  $  or $ \cledec(\cdot) $, respectively.  Finally,  \A~ outputs a bit $b'  $.
%		\item [--] At the end of the experiment, \A~ outputs a bit $b'  $.
%		, and wins if $ b'=b $.  
		%		\item[--]  If $  1 \leftarrow {\texttt{SGN.Ver}(m^*, \sigma^*, pk)} $ and $ m^* $ was not queried to $ \texttt{SGN.Sig}(\cdot) $, return 1, else, return 0.
		
	\end{itemize}

%	 We  define $\mathcal{O} (\cdot) =\texttt{PkReplace}(\cdot) $ if   $\mathcal{A}  = \mathcal{A}_I $ and  $ \mathcal{O} (\cdot)  = \texttt{PartialKeyExtract}(\cdot) $ if  $\mathcal{A} = \mathcal{A}_{II} $. 
		\vspace{-1mm}
	\A~wins the above experiment if     $ b=b' $ and the following conditions hold: (i) 	 $ ID^* $ was never submitted to  $ \oracle_{\texttt{Corrupt}} $.  (ii) 			 If $ \A=\A_I  $, $ ID^* $ was never submitted to $ \oracle_{\texttt{PartKey}}  $. (iii)	 If $ \A=\A_{II}  $, $ ID^* $ was never submitted to $ \oracle_{\texttt{SecKey}} $.   The IND-CLE-CCA  advantage of \A~ is  $    { \Pr [b=b']  \leq \frac{1}{2} + \epsilon }$, for a negligible $ \epsilon $.
\end{Definition}
\begin{Definition}\label{def:CertlessEUCMA} 
	The existential unforgeability  under chosen message attack	(EU-CLS-CMA) experiment \ExptCLS~ for  a  certificateless signature  $ {\cls} $  is defined as follows.  
		\begin{itemize}\normalfont
		\item [--] \Chall~runs \CLSKGCsetup($ 1^\kappa $) and returns $ mpk  $  \A. 
		\item[--] $ (ID^*,m^*,\sigma^*  )\leftarrow \Adv^{\oracle_{\texttt{PartKey}} ,\oracle_{\texttt{SecKey}},\oracle_{\texttt{Corrupt}},\oracle_{\texttt{Sign}}} (mpk) $
	\end{itemize}
 	\A~wins the above experiment if      $  1 \leftarrow {\sgnver(m^*, \sigma^*, ID)} $, and the following conditions hold: (i) $ ID^* $ was never submitted to  $ \oracle_{\texttt{Corrupt}} $.  (ii)  If $ \A=\A_I  $, $ ID^* $ was never submitted to $ \oracle_{\texttt{PartKey}}  $. (iii) If $ \A=\A_{II}  $, $ ID^* $ was never submitted to $ \oracle_{\texttt{SecKey}} $. The EU-CLS-CMA advantage of $  \Adv  $  is   $\Pr[\ExptCLS= 1]$
	\end{Definition}

%% file: Scheme.tex
 \section{Proposed Schemes} \label{sec:ProposedSchemes}
%We now present our identity-based and certificateless  schemes that are compatible with each other. 

%In this section, we propose our identity-based and certificateless encryption and signature schemes.  We also show how one can   derive a identity-based key exchange which  can be trivially extended to certificateless setting as well.  Finally, we will provide a discussion on how our identity-based and certificateless schemes are compatible. 

%
%\NOTE{We are proposing cryptosystems id based and certificateless. as an example we give enc and certless what is a cryptosystem? Key exchge }
%
%
%In the following, we will discus the main idea behind both of our proposed systems. 

\subsection{Proposed Identity-Based Cryptosystem}
 Most of pairing-free IDB schemes rely on the classical signatures (e.g., \cite{Schnorr91}) in their key generation to provide implicit certification. The use of such signatures to construct IDB schemes usually require several expensive operations (e.g., scalar multiplication), and therefore may incur a non-negligible computation overhead.  To reduce this cost, we exploit the message encoding  technique and subset resilient functions (similar to  \cite{HORS_BetterthanBiBa02}) along with the exponent product of powers  property  to generate keys. This permits an improved efficiency for both the PKG and user since it only  requires a hash call and a few point additions.   
 
  Our IDB schemes use similar \IBESetup~and \IBEExtract~functions whose key steps are outlined as follows. In the  \IBESetup, the PKG selects $ t $   values $ v_i \gets \ZZ_{q}$, and computes their commitments  as $ V_i \gets v_iP \mod p $, for $ i =1,\dots, t $, it then sets the master secret key $ msk \gets (v_1,\dots,v_t) $ and the system-wide public key $ mpk\gets (V_1,\dots,V_t) $.   This is similar to the scheme in \cite{HORS_BetterthanBiBa02}, where   EC scalar multiplication is used as the  one-way function. In   \IBEExtract, the PKG picks a nonce $ \beta\gets \ZZ_q $ and computes its commitments $ Q\gets \beta P \mod p $.   The PKG then  derives indexes   $(j_1,\ldots,j_k) \as \hash_1(ID,Q)$,  which select $k$-out-of-$t$ elements from the master secret key $ v_{j_i} $ for $ i=1,\dots,k $. Note that  $ Q $ is implicitly authenticated by being included in input of $ \hash_1(\cdot) $,  this is similar to the technique used in other pairing-free identity-based and certificateless systems \cite{DBLP:conf/africacrypt/GalindoG09,BaekCertLessPairingFree}. In  Steps 3-4, unlike the scheme in \cite{HORS_BetterthanBiBa02}, where secret keys are exposed, we use the additive homomorphic property in the exponent to mask the one-time signature $y$ (Step 3) via the nonce $ \beta $ (in line with~ \cite{Tachyon,ARIS}).  The PKG will then sends $ (x,Q) $ to the user via a secure channel.

\begin{algorithm*}[t]
	
	\small
	        \begin{minipage}{.93\linewidth}
		
	%	\caption{Basic  Optimal Static Key Exchange with $k$-Compromise Limit}
	\caption{Identity-Based    Encryption} \label{alg:IBE}
 
	\begin{multicols}{2}
		\begin{algorithmic}[1]
			\Statex\underline{$(msk,params)\gets $ \IBESetup$ (1^\kappa) $}:  
			\State Select primes $ p $ and $ q $  and $(t,k) \in \mathbb{N}$  where $ t>>k  $. 
			\For{$i = 1, \ldots, t$}
			
			\State  $v_i \Rq,~V_i \as {v_i}P \bmod p $
			\EndFor 
			
			\State \Return  $ msk \gets (v_1,\dots,v_t) $,  $ mpk \gets  (V_1,\dots,V_t)  $ and $ params \gets (\hash_1,\hash_2,\hash_3,\hash_4,p,q, k,t,mpk ) $
		\end{algorithmic}
		\algrule
		%	\vspace{5mm}
		
		\begin{algorithmic}[1]
			\Statex \underline{$ (w, Q) \gets $\IBEExtract$ (ID,U,msk) $}:  
			
			\State $ \beta \Ra \ZZ_q$, $ Q \gets \beta P  \bmod p$
			%		\State $ Q= U+W \bmod p$
			
			\State  $(j_1,\ldots,j_k) \as \hash_1(ID,Q)$ where for all $i = 1, \ldots, t$,  $ 1 < j_i <  |t|$
			
			\State  $y \as \sum_{i=1}^{k}v_{j_i} \bmod q$
			%		\State $Y \as \sum_{i=1}^{k}V_{j_i}  \bmod p$
			
			\State $x \gets y+\beta \bmod q $
			\State \Return $      (x,Q)$ 
		\end{algorithmic}

		\begin{algorithmic}[1]
			\Statex $\underline{c \leftarrow  \IBEEnc(m,ID_a, Q_a)}$:  Bob encrypts message $m \in \{0,1\}^n$.%Bob $(ID_b, Q_b)$ encrypts message $m \in \{0,1\}^n$  (for some $ n\in \mathbb{N} $) using the public key of Alice  $(ID_a, Q_a)$ as follows
			\vspace{3pt} 
			%	\State $r\xleftarrow{ \$ } \ZZ_{q} $, $R \leftarrow {r_b} P\bmod p $
			\State $\sigma \xleftarrow{ \$ }  \{0,1\}^n$, $r \leftarrow \hash_2(\sigma,m ),  R \leftarrow {r} P\bmod p$
			
			%		\State $C_a \leftarrow M_a + Q_a \bmod p $
			\State $(j_1,\ldots,j_k) \as \hash_1(ID_a,Q_a)$, $Y_a \as \sum_{i=1}^{k}V_{j_i}  \bmod p $
 
			\State   $u\gets \hash_3 (r(Y_a+Q_a)\bmod p) \oplus \sigma $,  $v\gets \hash_4 (\sigma ) \oplus m $  
			
			%	\State $d \leftarrow \texttt{MAC}_{k_{MAC}}(c)$
			\State \Return  $ c =(R,u,v) $ 
		\end{algorithmic} 
		\algrule
		% 	\vspace{2mm}
		
		\begin{algorithmic}[1]
			\Statex $\underline{m\leftarrow  \IBEDec(x_a,  c)}$: Alice decrypts the ciphertext $ c $.
			\vspace{3pt}

			\State   $\sigma' \gets  \hash_3 (x_aR\bmod p) \oplus u$
			\State  $m' \gets  v \oplus \hash_4(\sigma') $,$ r'\gets \hash_2(\sigma',m) $
			\If { $ r'P \pmod p =R $}\EndIf \Return  $ m' $
			\State \textbf{else} \Return $ \perp $
%			\State \textbf{else} \Return $ \perp $
		\end{algorithmic}
	\end{multicols}
     \end{minipage}
\end{algorithm*}

%=============================================== IBE ENDDD

%==============================================================IBS   
%\vspace{-1mm}
\begin{algorithm*}[t!]
\small
	\begin{minipage}{.93\linewidth}
		%	\caption{Basic  Optimal Static Key Exchange with $k$-Compromise Limit}
		\caption{Identity-Based    Signature} \label{alg:IBS}
		\vspace{-3mm}
		\begin{multicols}{2}
			\begin{algorithmic}[1]
				\Statex\underline{$(msk,params)\gets \IBSSetup$  $ (1^\kappa) $}:  Description identical to $ \IBESetup $ in Algorithm \ref{alg:IBE}, except that only the description of $\hash_1 $ and $ \hash_5 $ is included in $ params $.
				%			\State Select primes $ p $ and $ q $  and $(k,t)$  as in Definition 3.2.
				%			\For{$i = 1, \ldots, t$}
				%			
				%			\State  $v_i \Rq,~V_i \as {v_i}P \bmod p $
				%			\EndFor 
				%			
				%			\State \Return  $ msk \gets (v_1,\dots,v_t) $,  $ mpk \gets  (V_1,\dots,V_t)  $ and $ params \gets (\hash_1,\hash_2,\hash_3,\hash_4,p,q, k,t,mpk ) $
			\end{algorithmic}
			
			\algrule
			
			\begin{algorithmic}[1]
				\Statex \underline{$ (w, Q) \gets \IBSExtract $$ (ID,U,msk) $}:  As in  $ \IBEExtract $ in Algorithm \ref{alg:IBE}.
				
				%			\State $ \beta \Ra \ZZ_q$, $ Q \gets \beta P  \bmod p$
				%			%		\State $ Q= U+W \bmod p$
				%			
				%			\State  $(j_1,\ldots,j_k) \as \hash_1(ID,Q)$ where for all $i = 1, \ldots, t$,  $ 1 < j_i <  |t|$
				%			
				%			\State  $y \as \sum_{i=1}^{k}v_{j_i} \bmod q$
				%			%		\State $Y \as \sum_{i=1}^{k}V_{j_i}  \bmod p$
				%			
				%			\State $w \gets y+\beta \bmod q $
				%			\State \Return $      (w,Q)$ 
			\end{algorithmic}

			\algrule
			%	\begin{algorithmic}[1]
			%		\Statex\underline{$ x\gets $ \CLEUserKeyGen$ (params,\alpha,w) $}:  
			%		%			\State  $ \nu \gets \ZZ_{q} $
			%		%			\State $ Q\gets \nu P+  T \bmod p $
			%		\State $(j_1,\ldots,j_k) \as \hash_1(ID,Q)$, $Y \as \sum_{i=1}^{k}V_{j_i}  \bmod p $
			%		\State $ W' \gets Q-U\mod p  $, $ W'':=wP-Y \bmod p  $
			%		\If { $ W'=W'' $}      \Return   $ x \gets w+\alpha   \mod q$ \EndIf     
			%		\State 	\textbf{else}   \Return   $ \perp$
			%
			%	\end{algorithmic}
			%	\algrule
			\begin{algorithmic}[1]
				\Statex $\underline{(s,e)\leftarrow  \IBSSign(m,x_a)}$: Alice $ ID_a $ signs message $m$. 
				\vspace{3pt}
				\State $r \xleftarrow{ \$ } \ZZ_{q}$, $R \leftarrow {r} P\bmod p $
				\State $e \leftarrow \hash_5(m,R)$
				\State $s \leftarrow r- e \cdot x_a \bmod q$
				\State \Return $(s,e)$
			\end{algorithmic}
			
			\algrule
			\begin{algorithmic}[1]
				\Statex $\underline{\{0,1\}\leftarrow  \IBSVerify(m, ID_a,Q_a,\langle s,e \rangle)}$: Bob verifies the signature  $(s,e)$.
				% 				\vspace{3pt}
				\State $(j_1,\ldots,j_k) \as \hash_1(ID_a,Q_a) $
				\State $Y_a \as \sum_{i=1}^{k}V_{j_i}  \bmod p  $
				\State $R' \as sP + e(Y_a + Q_a) \bmod p$
				\If{$e =\hash_5(m,R')$} \Return 1  
				\Else~\Return 0 
				\EndIf
				
			\end{algorithmic}
			
			\columnbreak
		\end{multicols}
	\end{minipage}
\end{algorithm*}\setlength{\textfloatsep}{2pt}

\noindent\textbf{Identity-Based Encryption Scheme:} In  \IBEEnc~(\autoref{alg:IBE},  Step 2),   the indexes obtained from $ \hash_1 $ are used to retrieve the components $ V_{j_i}  $ from the system-wide public key $ mpk $. The input of $ \hash_3 $ is the ephemeral key, which  given the ciphertext $ c =(R,u,v) $, can be recomputed by the receiver  in the \IBEDec~algorithm.   $ \sigma $ and $ r $ are computed in-line with the transformation proposed in \cite{DBLP:conf/crypto/FujisakiO99}.

%The \IBEEnc~algorithm in \autoref{alg:IBE} is obtained by  making use of the method proposed in  \cite{DBLP:conf/crypto/FujisakiO99} and therefore, shares many similarities in how $ \sigma $ and $ r $ are computed. In Step 2,  the indexes obtained from $ \hash_1 $ are used to retrieve the components $ V_{j_i}  $ from the system-wide public key $ mpk $. Note that the input of $ \hash_3 $ is the ephemeral key which  given the ciphertext $ c =(R,u,v) $, can be recomputed by the receiver  in the \IBEDec~algorithm.

\noindent\textbf{Identity-Based Signature Scheme:}    In   \IBSVerify, the public key of the user $Y_a$ is computed from $V_{j_i} \in mpk$ via the indexes retrieved from the output of $ \hash_1 $. The key generation is as in \autoref{alg:IBE}. The rest of the signing and verification steps are akin to   Schnorr signatures~\cite{Schnorr91}. 
%
%The signer computes a nonce $ r $ and computes its commitment, which is also included as an input to $\hash_5 $. The nonce will then be used to mask the product of the private key and the hash of the message.  In   \IBSVerify, the public of the user is computed by using the indexes retrieved from the output of $ \hash_1 $ and adding the corresponding elements in $ mpk $, i.e., $ V_{j_i} $. The verifier   outputs $ 1  $ (accept) if the recomputed $ R' $ (in Step 3), matches the one computed by the signer (Step 4). 

 \noindent\textbf{Identity-Based Key Exchange Scheme:}   For the key exchange scheme, we run $ \IBESetup $ and then let both parties, Alice and Bob,  obtain   $ (x_A,Q_A )$ and $ (x_B,Q_B) $ via the \IBEExtract~algorithm, respectively. Alice then picks $ z_A \Ra \ZZ_{q}  $, computes its commitment $  M_A\gets z_A P\bmod p    $, and sends $ (M_A,Q_A) $ to Bob. Bob does the same and sends  $(M_B,Q_B) $ to Alice. Alice then computes $(j_1,\ldots,j_k) \as \hash_1(ID_b,Q_b)$ and  	$Y_b \as \sum_{i=1}^{k}V_{j_i}  \bmod p $  and outputs the shared secret key as  $K_a \leftarrow   x_a(Y_b+Q_b)+ z_aM_b \bmod p$. Bob   works similarly, and outputs the shared key as $K_b \leftarrow x_b(Y_a+Q_a)+z_bM_a\bmod p$.

\begin{algorithm*}[h]

	%	\caption{Basic  Optimal Static Key Exchange with $k$-Compromise Limit}
	\caption{Certificateless    Encryption} \label{alg:CLENC}
	\small
		  \begin{minipage}{.93\linewidth}
		  \begin{multicols}{2}
	\begin{algorithmic}[1]
		\Statex\underline{$(msk,params)\gets $ \CLEKGCsetup$ (1^\kappa) $}:    As in  $ \IBESetup $  in Alg. \ref{alg:IBE}. 
%		\State Select primes $ p $ and $ q $  and $(t,k) \in \mathbb{N}$  where $ t>>k  $. 
%		\For{$i = 1, \ldots, t$}
% 
%		\State  $v_i \Rq,~V_i \as {v_i}P \bmod p $
%		\EndFor 
%		
%		\State \Return  $ msk \gets (v_1,\dots,v_t) $,  $ mpk \gets  (V_1,\dots,V_t)  $ and $ params \gets (\hash_1,\hash_2,\hash_3,\hash_4,p,q, k,t,mpk ) $
	\end{algorithmic}

 	\algrule
	\begin{algorithmic}[1]
		\Statex\underline{$ (\alpha,U)\gets $\CLEUserSetup$ (\cdot) $}:  %This algorithm is  run by the user $ ID $. 
		\State  $ \alpha \Ra \ZZ_{q} ,U \gets \alpha P \bmod p $
		\State \Return  $ (\alpha,  U )$
	\end{algorithmic}
 
	\algrule

	\begin{algorithmic}[1]
		\Statex \underline{$ (w, Q) \gets $\CLEPartialKeyGen$ (ID,U,msk) $}:  %After user $ ID $ authenticates itself to the KGC and provides it with $ (ID,U) $, KGC executes: %this algorithm is run by the KGC.
 
		\State $ \beta \Ra \ZZ_q$, $ W \gets \beta P  \bmod p$
		\State $ Q= U+W \bmod p$
		
		\State  $(j_1,\ldots,j_k) \as \hash_1(ID,Q)$ where for all $i = 1, \ldots, t$,  $ 1 < j_i <  |t|$
 
		\State  $y \as \sum_{i=1}^{k}v_{j_i} \bmod q$
		%		\State $Y \as \sum_{i=1}^{k}V_{j_i}  \bmod p$
		
		\State $w \gets y+\beta \bmod q $
		\State \Return $      (w,Q)$ 
	\end{algorithmic}
 
 \columnbreak
 
	\begin{algorithmic}[1]
		 \Statex\underline{$ x\gets $ \CLEUserKeyGen$ (w,\alpha) $}:
		%			\State  $ \nu \gets \ZZ_{q} $
		%			\State $ Q\gets \nu P+  T \bmod p $
		\State $(j_1,\ldots,j_k) \as \hash_1(ID,Q)$, $Y \as \sum_{i=1}^{k}V_{j_i}  \bmod p $
		\State $ W' \gets Q-U\mod p  $, $ W'':=wP-Y \bmod p  $
		 \If { $ W'=W'' $}      \Return   $ x \gets w+\alpha   \mod q$  \textbf{else}   \Return   $ \perp$\EndIf     
%	\State 	\textbf{else}   \Return   $ \perp$

	\end{algorithmic}	    
\algrule

\begin{algorithmic}[1]
	\Statex $\underline{c \leftarrow  \cleenc(m,ID_a, Q_a)}$: Bob encrypts message $m \in \{0,1\}^n$.
%	\vspace{3pt} 
%	\State $r\xleftarrow{ \$ } \ZZ_{q} $, $R \leftarrow {r_b} P\bmod p $
	\State $\sigma \xleftarrow{ \$ }  \{0,1\}^n$, $r \leftarrow \hash_2(\sigma,m ),  R \leftarrow {r} P\bmod p$
	
	%		\State $C_a \leftarrow M_a + Q_a \bmod p $
	\State $(j_1,\ldots,j_k) \as \hash_1(ID_a,Q_a)$, $Y_a \as \sum_{i=1}^{k}V_{j_i}  \bmod p $
 
%	\State  $S_{ab} \leftarrow   r_b(Y_a+Q_a) \bmod P$ \RB{This can directly be used in the next line }
%	\State $(k_{enc}, k_{MAC}) \leftarrow \texttt{KDF}(\mathbf{S})$
	\State   $u\gets \hash_3 (r(Y_a+Q_a)\bmod p) \oplus \sigma $,  $v\gets \hash_4 (\sigma ) \oplus m $  
 
%	\State $d \leftarrow \texttt{MAC}_{k_{MAC}}(c)$
 	\State \Return  $ c =(R,u,v) $ 
\end{algorithmic}
\algrule

\begin{algorithmic}[1]
	\Statex $\underline{m\leftarrow  \cledec(x_a,  c)}$:   Alice decrypts the ciphertext $ c $.  
%	\vspace{3pt}
 
%	\State $S'_{ab}\leftarrow  x_aR_b\bmod p$
%	\State $(k_{enc}, k_{MAC}) \leftarrow \texttt{KDF}(SS')$
% 	\State If $d \neq  \texttt{MAC}_{k_{MAC}}(c)$ return {\em invalid}
%	\State $m' \as S'_{ab} +c$
	\State   $\sigma' \gets  \hash_3 (x_aR\bmod p) \oplus u$ 
	\State  $m' \gets  v \oplus \hash_4(\sigma') $,$ r'\gets \hash_2(\sigma',m) $
		 \If { $ r'P \pmod p =R $} \EndIf \Return  $ m' $
	\State \textbf{else} \Return $ \perp $
\end{algorithmic}
\end{multicols}
\end{minipage}
\end{algorithm*} 
 
%========3========
 
%==============================================================CLS  Begin
\begin{algorithm*}[t!]
 		\caption{Certificateless Digital Signature}\label{alg:CLS}
	\small
		  \begin{minipage}{.93\linewidth}

	%	\hspace{5pt}
 
		\begin{multicols}{2}
			
			\begin{algorithmic}[1]
				\Statex\underline{$(msk,params)\gets $ \CLSKGCsetup$ (1^\kappa) $}: As in $ \CLEKGCsetup $ in Alg. \ref{alg:CLENC}, except that  $\hash_1 $ and $ \hash_5 $ are in $ params $.   
				%		\State Select primes $ p $ and $ q $  and $(k,t)$  as in Definition 3.2.
				%		\For{$i = 1, \ldots, t$}
				%		\State  $v_i \Rq,~V_i \as {v_i}P \bmod p $
				%		\EndFor 
				%		\State \Return  $ msk \gets (v_1,\dots,v_t) $,  $ mpk \gets  (V_1,\dots,V_t)  $ and $ params \gets (p,q, k,t,mpk ) $
			\end{algorithmic}
			\algrule
			
			\begin{algorithmic}[1]
				\Statex\underline{$ (\alpha,U)\gets $\CLSUserSetup$ (params) $}:  As in \linebreak  $ \CLEUserSetup $ in Alg. \ref{alg:CLENC}.
				%		\State  $ \alpha \Ra \ZZ_{q} ,U \gets \alpha P \bmod p $
				%		\State \Return  $ (\alpha,  U )$
			\end{algorithmic}
			\algrule

			\begin{algorithmic}[1]
				\Statex \underline{$ (w, Q) \gets $\CLSPartialKeyGen$ (ID,U,msk) $}:  As in  $ \CLEPartialKeyGen $ in Alg. \ref{alg:CLENC}.
				%		\vspace{3pt}
				%		\State $ \beta \Ra \ZZ_q$, $ W \gets \beta P  \bmod p$
				%		\State $ Q\gets  U+W \bmod p$
				%		
				%		\State  $(j_1,\ldots,j_k) \as \hash_1(ID,Q)$ where for all $i = 1, \ldots, t$,  $ 1 < j_i <  |t|$
				%		\State  $y \as \sum_{i=1}^{k}v_{j_i} \bmod q$
				%		%		\State $Y \as \sum_{i=1}^{k}V_{j_i}  \bmod p$
				%		
				%		\State $w \gets y+\beta \bmod q $
				%		\State \Return $      (w,Q)$ 
			\end{algorithmic}
			
			\algrule

			\begin{algorithmic}[1]
				\Statex\underline{$ x\gets $ \CLSUserKeyGen$ (params,\alpha,w) $}: As in  $ \CLEUserKeyGen $ in Alg. \ref{alg:CLENC}.
				%%		%			\State  $ \nu \gets \ZZ_{q} $
				%%		%			\State $ Q\gets \nu P+  T \bmod p $
				% 	    \State $(j_1,\ldots,j_k) \as \hash_1(ID,Q)$, $Y \as \sum_{i=1}^{k}V_{j_i}  \bmod p $
				% 		 \State $ W' \gets Q-U\mod p  $, $ W'\gets wP-Y \bmod p  $
				% 		\If { $ W'=W'' $}      \Return   $ x \gets w+\alpha   \mod q$ \EndIf     
				% 		\State 	\textbf{else}   \Return   $ \perp$
				%		

			\end{algorithmic}

			%	\begin{algorithmic}[1]
			%		\Statex $\underline{(x,Q) \leftarrow \texttt{KeyGen}(1^{\kappa})}$:  %Given $1^{\kappa}$ as the input
			%		\vspace{3pt}
			%		\State $(msk, params)\gets  $ \syssetup $ (1^\kappa) $
			%		\State $ (\alpha,U)\gets $ \usrsetup$ (params) $
			%		\State $ (w,Q)\gets $\extract$ (params,ID,U) $
			%		\State $ x\gets $ \keygen$ (params,\alpha,w) $
			%		\State \Return $      (x,Q)$ 
			%	
	\algrule
			
			\begin{algorithmic}[1]
				\Statex $\underline{(s,e)\leftarrow  \CLSSign(m,x_a)}$: Alice $ ID_a $ signs message $m$.  %Alice $(ID_a,Q_a)$ signs message $m\in  \{0,1\}^*$ using its private key:
	%			\vspace{3pt}
				\State $r \xleftarrow{ \$ } \ZZ_{q}$, $R \leftarrow {r} P\bmod p $
				\State $e \leftarrow \hash_5(m,R)$
				\State $s \leftarrow r- e \cdot x_a \bmod q$
				\State \Return $(s,e)$
			\end{algorithmic}
			\algrule
			
			\begin{algorithmic}[1]
				\Statex $\underline{\{0,1\}\leftarrow  \CLSVerify(m, Q_a,\langle s,e \rangle)}$:  Bob verifies the signature  $(s,e)$.
%				\vspace{3pt}
				\State $(j_1,\ldots,j_k) \as \hash_1(ID_a,Q_a) $
				\State $Y_a \as \sum_{i=1}^{k}V_{j_i}  \bmod p  $
				\State $R' \as sP + e(Y_a + Q_a) \bmod p$
				\If{$e =\hash_5(m,R')$} \Return 1  
				\Else~\Return 0 
				\EndIf
				
			\end{algorithmic}
		\end{multicols}
	\end{minipage}
\end{algorithm*} 
%==============================================================CLS  END
 
\subsection{Proposed Certificateless Cryptosystem}
 For our CL schemes to achieve the same trust level (Level 3) \cite{SelfCertGIRAULTOriginal91} on the third party (KGC), as in traditional PKI, we use the binding method \cite{PatersonCertLess} in the \CLEPartialKeyGen~and \CLSPartialKeyGen~algorithms. Note that the same  secure channel which is used for user authentication  (e.g., SSL/TLS), can be used to send  the user commitment  $ U $  to the KGC.  This permits an implicit certification of  $ U, $ and therefore any changes  of  $ U, $ will falsify the private key.   

The \CLEKGCsetup~algorithm is as in \IBESetup~in \autoref{alg:IBE}.  The \linebreak \CLEPartialKeyGen~algorithm is similar to the \IBEExtract~in \autoref{alg:IBE}, with the difference that the user commitment $ U $ is used to compute $ Q $. In  \CLEUserKeyGen,     the correctness  of the partial private key is checked first before the  private key $ x $ is computed.    

%} . However, 
%
%In our certificateless schemes, to achieve the same trust level (Level 3) \cite{SelfCertGIRAULTOriginal91} on the third party (KGC) as in traditional PKI, we use the binding method. As aforementioned, using this  method also allows us to achieve a more relaxed security model. 
%
%We now outline the main idea behind our certificateless cryptosystems. The \CLEKGCsetup~algorithm is identical to   \IBESetup~in \autoref{alg:IBE}. In our certificateless schemes, to achieve the same trust level (Level 3) \cite{SelfCertGIRAULTOriginal91} on the third party (KGC) as in traditional PKI, we use the binding method. As aforementioned, using this  method also allows us to achieve a more relaxed security model. 
%
%
%In our certificateless schemes, we use the binding method \cite{PatersonCertLess} in the \linebreak \CLEPartialKeyGen~and \CLSPartialKeyGen~algorithms.  Note that in order for the user to authenticate itself to the KGC, it needs to establish a secure channel. We believe that using the same secure channel, the user can send the commitment of its secret key (i.e., $ U $) to the KGC.  Note that using this technique, the commitment value is implicitly certified and any changes to this value will completely falsify the partial private key and private key.   

\noindent\textbf{Certificateless  Encryption Scheme:}  Note that the \cleenc~and \cledec~algorithms are identical to \IBEEnc~and \IBEDec~algorithms in \autoref{alg:IBE}. 

 \noindent\textbf{Certificateless  Signature  Scheme:} The setup and key generation algorithms are as in \autoref{alg:CLENC}, and the \CLSSign~and \CLSVerify~algorithms are as in \IBSSign~and \IBSVerify~in \autoref{alg:IBS}, respectively. 
 
\noindent\textbf{Certificateless  Key Exchange Scheme:} Given the compatibility of our IDB and CL schemes, after the initial algorithms (system setup and key generation) take place as in \autoref{alg:CLENC}, the CL key exchange will be identical to the one proposed in  the identity-based key exchange scheme above. 

  \subsection{Compatibility of Identity-Based and Certificateless Schemes}

 In our CL schemes, we utilize the additive homomorphic property of the exponents  (i.e., $ w $) when the KGC includes the addition of commitments ($ W  $and $ U $) in the $ \hash_1 $. After receiving  $ w $, the user exploits the  homomorphic property  to modify the key without falsifying it and obtain $ x $. For instance, we observed that our counterparts (e.g.,  \cite{BonehID-based,PatersonCertLess}) do not  offer such a compatibility, since the partial private key is the KGC's commitment to the  (hash of)  user  identity, without a homomorphic property. Moreover, the KGC does not output any auxiliary value to incorporate the user commitment  with it. 

As shown above, our IDB and CL schemes are compatible,  thanks to the special design of their key generation algorithms  (i.e., $ \texttt{Extract} $ in IDB, $ \texttt{UserSetup} $ and $ \texttt{PartKeyGen} $ in CL).  Therefore, after the users computed/obtained their keys from the third party, the interface of the main cryptographic functions (e.g., encrypt, decrypt, sign, etc.) are identical in both systems, therefore, the users can communicate with uses in different domains seamlessly. For instance,  ciphertext $ c=(R,u,v) $ outputted by the  \cleenc~in \autoref{alg:IBE},  can be decrypted by a user in the identity-based setting by  the \IBEDec~algorithm in \autoref{alg:IBE}. This also applies to the signature and the key exchange schemes proposed above.

%% file: Sec-new.tex
 \section{Security Analysis}\label{Sec:SecAnalysis}
 \vspace{-1mm}
 
%  \RB{Mention that given the sec of certless, the sec of IDbased is trivial\\ need to consider the decline of decryption query for users which pk is replaced\\ }
%\subsubsection{Two Public Key Encryption Schemes:}
\newcommand{\hybrid}{\ensuremath {\texttt{HybridPub}}{\xspace}}
\newcommand{\basic}{\ensuremath {\texttt{BasicPub}}{\xspace}}
\newcommand{\KeyGen}{\ensuremath {\texttt{KeyGen}}{\xspace}}
\newcommand{\Encrypt}{\ensuremath {\texttt{Encrypt}}{\xspace}}
\newcommand{\Decrypt}{\ensuremath {\texttt{Decrypt}}{\xspace}}
\newcommand{\OWE}{\ensuremath {\text{OWE}}{\xspace}}
\newcommand{\TI}{{{Type-I~}}{\xspace}}

\newcommand{\TII}{{{Type-II~}}{\xspace}}

\newcommand{\B}{\ensuremath {\mathcal{B}}{\xspace}}

\newcommand{\ChallI}{\ensuremath {\Chall_I}{\xspace}}
\newcommand{\ChallII}{\ensuremath {\Chall_{II}}{\xspace}}

\newcommand{\event}{\ensuremath {\small\texttt{E}_1}{\xspace}}
 \newcommand{\eventhashII}{\ensuremath{\small\texttt{E}_{\texttt{H}_2}}{\xspace}}
  \newcommand{\eventhashIII}{\ensuremath{\small\texttt{E}_{\texttt{H}_3}}{\xspace}}
    \newcommand{\eventDec}{\ensuremath{\small\texttt{Dec}_{\texttt{err}}}{\xspace}}

	\begin{theorem}\label{thm:IND-CLE-CCA}
		If an adversary  $ \A_I $ can break the IND-CLE-CCA security of the encryption scheme proposed in Algorithm \ref{alg:CLENC} after $  q_{\hash_i}$ queries to random oracles $ \hash_i $ for $ i \in \{1,2,3,4\}  $,    $q_{D} $ queries to the decryption oracle and $ q_{\sk} $ to the private key extraction oracle with probability $ \epsilon $. There exists another algorithm \Chall~ that runs $ \A_I $ as subroutine and breaks a random instance of the CDH problem ($ P,aP,bP $) with probability $ \epsilon' $ where: $ \epsilon' > \frac{1}{q_{\hash_3}}\big(\frac{2\epsilon}{\mathbf{e}(q_{\sk}+1)} - \frac{q_{\hash_2}}{2^n}-\frac{ q_D(q_{\hash_2}+1)}{2^n} -\frac{ 2q_D}{p}\big) $.
	\end{theorem}

\begin{proof}
	Our proof technique is similar to the one in \cite{BaekCertLessPairingFree}.
	% We let   $ \A_I $ be a Type-I IND-CLE-CCA adversary against the scheme proposed in Algorithm \ref{alg:CLENC}, we then show that one can build another algorithm \Chall~that runs in polynomial time and uses $ \A_I $ as  subroutine to solve an instance of the the CDH problem. 
	\Chall~\emph{simulates} the real environment for  $ \A_I $.  It knows  the $ t  $ secret values $ v_i's $ in the scheme, and tries to embed a random instance of the CDH problem $(P,aP,bP)  $. \Chall~sets $ aP $ as a part of the target user's ($ ID^* $)   public key  (i.e., $ Q_{ID^*} \gets aP$) and $ bP $ as a part of the challenge ciphertext (i.e., $R^*\gets bP) $.  \Chall~uses four lists, namely  \ListHI, \ListHII, \ListHIII, and \ListHIV, to keep track of the random oracle responses and following the IND-CLE-CCA experiment $ \mathit{Expt}^{\textit{IND\mbox{-}CLE\mbox{-}CCA}}_{\A} $   (Definition \ref{def:CertlessINDCCA}),   \Chall~responds to \AI~queries as follows. 
	
	\noindent \emph{Queries to} $ \hash_1(ID_i,Q_i) $:  If the entry $ (\langle ID_i,Q_i\rangle,h_{1,i})   $ exists in \ListHI, \Chall~returns $ h_{1,i} $,  otherwise, it chooses  $ h_{1,i}\Ra \gamma  $, and inserts $ (\langle ID_i,Q_i\rangle,h_{1,i})   $ in \ListHI. 
	
	\noindent \emph{Queries to} $ \hash_2(\sigma_i,m_i) $:  If the entry $ (\langle \sigma_i,m_i\rangle,h_{2,i})   $ exists in \ListHII, \Chall~returns $ h_{2,i} $,  otherwise, it chooses  $ h_{2,i} \Ra \ZZ_q  $, and inserts  $ (\langle \sigma_i,m_i\rangle,h_{2,i})   $ in \ListHII. 
	
	\noindent \emph{Queries to} $ \hash_3(K_i) $:  If the entry $ (K_i,h_{3,i})   $ exists in \ListHIII, \Chall~returns $ h_{3,i} $,  otherwise, it chooses  $ h_{3,i} \Ra \{0,1\}^n  $, and inserts  $  (K_i,h_{3,i})   $  in \ListHIII. 
	
	\noindent \emph{Queries to} $ \hash_4(\sigma_i) $:  If the entry $ (\sigma_i,h_{4,i})   $ exists in \ListHIV, \Chall~returns $ h_{4,i} $,  otherwise, it chooses  $ h_{4,i} \Ra \{0,1\}^n  $, and inserts $ (\sigma_i,h_{4,i})   $  in \ListHIV. 
	
	\noindent \emph{Public key request}:   Upon receiving a public key request on $ ID_i,$ \Chall~works as follows. If $ (\langle ID_i, U_i,Q_i\rangle , \zeta_i) $ exists in $ \texttt{List}_\pk $, then it returns $ (ID_i, U_i,Q_i) $. Else, it flips a fair coin  where $ \Pr[\zeta = 0] = \delta $, and works as follows ($ \delta $ will be determined later in the proof).   	If $ \zeta = 0 $, it runs the partial key extraction oracle below first, update  $ \texttt{List}_\pk $ and then output  $(ID_i, U_i,Q_i)  $.  If $ \zeta =1 $, pick $ t\Ra \ZZ_{q}  $, set $ Q_i\gets aP \mod p $, adds $ (ID_i, U_i,\langle \perp, Q_i  \rangle )   $ to  $ \texttt{List}_{\texttt{PartialSK}} $  and  adds $ (\langle ID_i, U_i,Q_i\rangle , \zeta_i) $ to $ \texttt{List}_\pk $, before outputting   $ (ID_i, U_i,Q_i)  $. 
	
	\noindent \emph{Partial key extraction}: Upon receiving a partial key extraction query   on $ (ID_i,U_i) $,  \Chall~works as follow:
	\vspace{-2mm}
	\begin{itemize}
		
		\item If  $ (ID_i,U_i, \langle w_i, Q_i  \rangle )   \in \texttt{List}_{\texttt{PartialSK}} $, return $ (w_i, Q_i )  $.
		\item Else,
		\begin{itemize} 	
			\item  $ w_i \Ra \ZZ_{q} $,  $Z_i\gets  w_i P \mod p $, $  (j_1,\dots,j_k ) \Ra  [1,\dots, t] $, $ Q_i\gets Z_i-\sum_{i=1}^{k}V_{j_i} +U_i \mod p$.
			
			\item If $ (ID_i,Q_i,\dots )  \in \ListHI  $,  aborts. Else,  adds  $ (\langle ID_i,Q_i\rangle,h_{1,i})   $ to \ListHI, where $ h_{1,i} \gets (j_1,\dots,j_k )   $   and output the partial private key as $ (w_i,U_i,Q_i)$ after adding it to $ \texttt{List}_{\texttt{PartialSK}} $. 
		\end{itemize}   
	\end{itemize}

	%Secret key request
	\noindent \emph{Secret key request}: Upon receiving a secret key request   on $ ID_i $, \Chall~checks if  there exists a pair $ (ID_i,u_i,U_i) \in \texttt{List}_{\texttt{SecretKey}}  $, it returns $ u_i $. Otherwise,  selects  $ u_i\Ra \ZZ_{q} $, computes $ U_i\gets  u_iP\mod p $ and inserts $ (ID_i,u_i,U_i)  $ in $ \texttt{List}_{\texttt{SecretKey}}  $.

	% Private key request
	\noindent \emph{Private key request}:   To answer a private key request on $ (ID_i,U_i) $, \Chall~runs the public key request oracle above to get $ (\langle ID_i, U_i,Q_i\rangle , \zeta_i)  \in  \texttt{List}_\pk $ and finds $ (ID_i,u_i,U_i) $ in  $\texttt{List}_{\texttt{SecretKey}}  $ . If $ \zeta =0 $, finds $ ( (ID_i,U_i, \langle w_i, Q_i  \rangle )   \in \texttt{List}_{\texttt{PartialSK}} $ and returns $ w_i+u_i $ as the response. Otherwise, it  aborts.

	%Decryption 
	\noindent \emph{Decryption query}:  Upon receiving a decryption query on $ (ID_i,Q_i,c_i= \langle R_i,u_i,v_i\rangle) $, \Chall~works as follows. 
	
	\begin{itemize} 	
		\item   Searches $ \texttt{List}_\pk $ for an entry $ (\langle ID_i, U_i,Q_i\rangle , \zeta_i)  $. If $ \zeta = 0 $, works as follows. 
		\begin{itemize} 	
			\item  Searches  $ \texttt{List}_{\texttt{PartialSK}} $ for a tuple $(ID_i,U_i, \langle w_i, Q_i  \rangle )   $ and searches for  $ (ID, \langle w, Q  \rangle )   $   in $ \texttt{List}_{\texttt{PartialSK}} $, set $\sigma'\gets  \hash_3((w+\alpha)R \mod p) \oplus u $, $ m'=v\oplus\hash_4(\sigma') $, $ r':= \hash_2(\sigma',m)$. 
			\item Checks if $   R= r'P \mod p  $ holds, outputs $ m' $
		\end{itemize}   
		\item Else, if $ \zeta =1 $, works as follows. 
		
		\begin{itemize} 	
			\item   Runs the oracle for $ \hash_1 $ to get   $h_{1,i} $ (to compute the public key $ Y_i $) and checks lists   \ListHII, \ListHIII~and \ListHIV~for tuples $ (\langle \sigma_i,m_i\rangle,h_{2,i})   $,  $  (K_i,h_{3,i})    $, and $(\sigma_i,h_{4,i})  $,   such that $ R_i = h_{2,i}P \bmod p $, $ u = h_{3,i}  \oplus \sigma_i $ and $ v = h_{4,i}   \oplus m_i $  exists.  Checks if $ K_i = r_i(Y_i+Q_i) $ holds,  outputs $ m_i $, else, aborts. 
		\end{itemize}   
	\end{itemize}

	After the first round of queries, 	 \AI~outputs  $ ID^*$ and two messages $ m_0 $ and $ m_1 $ on which it wishes to be challenged on. We assume that $ ID^* $ has been already queried to $ \hash_1 $ and  was not submitted to the \emph{private key request oracle}. \Chall~checks $ (\langle ID^* ,U^* ,Q^* \rangle,\zeta) \in  \texttt{List}_{PK} $ if  $  \zeta =0$, it~aborts. Otherwise, it computes the challenge ciphertext as follows.   $ \beta^*\Ra \{0,1\} $, $ \sigma^* \Ra \{0,1\}^* $, $ u^* \gets \{0,1\}^{n} $, $ b \Ra \{0,1\} $.  $ R^* \gets aP $ (this implicitly implies that $ a =  \hash_2(\sigma^*,m_b)  $), $ \hash_3 (K_{ID^*}) \gets  u^*\oplus  \sigma^*$ and $ v^*\gets \hash_4(\sigma^*)\oplus m_b $. Return $ (R^*,u^*,v^*) $.
	
	\AI~initiates the second round of queries similar as above, with the restrictions defined in Definition \ref{def:CertlessEnc}.  When $ \AI $ outputs its decision bit $ b' $, \Chall~returns a set $ \Lambda = \{K_{i}-R_{i}^{y_i}, $where $ K_i $s  are the input queries to \hash$ _3 $$\} $.
	
	Notice that if \Chall~does not   abort, and $ \AI $ outputs its decision bit $ b' $, then the public key must have the $ Q_{ID^*}  = aP$, and given how the challenge ciphertext is formed (e.g., $ R^*=bP $), $ K_{ID^*} = y_{ID^*}abP $ should hold,  where $  y_{ID^*} $ is known to \Chall. Hence, the answer to a random instance of the the  CDH  problem $ (P,aP,bP) $, can be derived from examining the \AI's choice of public key and $ \hash_3  $ queries. 
	
	Here, we provide an indistinguishability argument for the above simulation. First we look at the simulation of the decryption algorithm. If $ \zeta =0 $, we can see that the simulation is perfect. For $ \zeta =1  $, an error might occur in the event that $ c_i $ is valid, but $ (\sigma_i,m_i)  $, $ K_i $, and $ \sigma_i $ were never queried to $ \hash_2,\hash_3, $ and $ \hash_4 $, respectively. For the first two hash functions, the probability that the $ c_i  $ is valid, given  a query to $ \hash_3 $ was never made, considers the query to $ \hash_2$ as well (not considering the checking phase in the simulation). Therefore, the probability that this could occur is $ \frac{q_{\hash_2}}{2^n} +\frac{1}{p}  $.  When considering $ \hash_4 $, this probability is $ \frac{1}{2^n} +\frac{1}{p}  $. Given the number of decryption queries $ q_D $, we have the probability of decryption error $  \frac{ q_D(q_{\hash_2}+1)}{2^n} +\frac{ 2q_D}{p}  $. 
	
	\Chall~will also fail in simulation during the partial key extraction queries if the entry $ (ID_i,Q_i, \dots) $ already exists in $ \ListHI $. This will happen with probability $ \frac{q_{\hash_1}}{2^\gamma}$. 
	
	The probability that~\Chall~does not abort in the simulation is $ \delta^{q_{\sk}}(1-\delta) $ which is maximized at $ \delta = 1-\frac{1}{q_{\sk}+1} $.  Therefore, the probability that $ \Chall $ does not abort is $ \frac{1}{\mathbf{e}(q_{\sk}+1)} $, where $ \mathbf {e} $ is the base of natural logarithm. Given the argument above, we know that if $ (\sigma^*,m_b) $, $ (K^*) $  were never  queried to $ \hash_2 $ and $ \hash_3 $ oracles, then \AI~cannot gain any distinguishing advantage more than $ \frac{1}{2} $. Given all the above arguments, the probability that $ K_{ID^* }$ has been queried to $ \hash_3 $ is $ \geq \frac{2\epsilon}{\mathbf{e}(q_{\sk}+1)} - \frac{q_{\hash_2}}{2^n}-\frac{ q_D(q_{\hash_2}+1)}{2^n} -\frac{ 2q_D}{p}$.
	
	Therefore, if the above probability occurs, \Chall~can solve the CDH problem by finding and computing  $ K_{ID^*} = y_{ID^*}abP $ from the list $ \Lambda $. Given the size of the list $ \Lambda $ (i.e., $ q_{\hash_3} $), the probability for \Chall~to be successful in solving CDH is:  $ \epsilon' > \frac{1}{q_{\hash_3}}\big(\frac{2\epsilon}{\mathbf{e}(q_{\sk}+1)} - \frac{q_{\hash_2}}{2^n}-\frac{ q_D(q_{\hash_2}+1)}{2^n} -\frac{ 2q_D}{p}\big)$
	%
	%
	%Lets define  $  \eventhashII$ be the event where $ (\sigma^*,m_b) $ was queried to $ \hash_2 $ oracle  and   $  \eventhashIII$ be the event where $ (K^*) $ was queried to $ \hash_3 $ oracle, and \eventDec be the event that ciphertext is valid and not query was made to $\hash_3$.  
\end{proof}

	\begin{theorem}\label{thm:IND-CLE-CCA-II}
	If an adversary  $ \A_{II} $ can break the IND-CLE-CCA security of the encryption scheme proposed in Algorithm \ref{alg:CLENC} after $  q_{\hash_i}$ queries to random oracles $ \hash_i $ for $ i \in \{1,2,3,4\}  $,    $q_{D} $ queries to the decryption oracle and $ q_{\sk} $ to the secret  key extraction oracle with probability $ \epsilon $. There exists another algorithm \Chall~that runs $ \AII$ as subroutine and breaks a random instance of the CDH problem ($ P,aP,bP $) with probability $ \epsilon' $ where: $ \epsilon' > \frac{1}{q_{\hash_3}}\big(\frac{2\epsilon}{\mathbf{e}(q_{\sk}+1)} - \frac{q_{\hash_2}}{2^n}-\frac{ q_D(q_{\hash_2}+1)}{2^n} -\frac{ 2q_D}{p}\big) $.
\end{theorem}
\begin{proof}[Sketch]  %	Due to the page limitation, the full proof will be provided in the full version of the paper. 
	Having access to  random oracles,  and by keeping lists similar to above, the challenger	\Chall~can simulate an indistinguishable environment for \AII~and respond to its queries similar to the above proof.  Note that following \autoref{def:CertlessINDCCA}, \AII~can query for the secret key of all the users, except for the target user $ ID^* $.  
	
	\Chall~knows  the $ t  $  private values $ v_i's $ in the scheme, and tries to embed a random instance of the CDH problem $(P,aP,bP)  $.   By flipping a fair coin, as in the {\em public key request} query above, \Chall~defines the probability to embed  $ aP $ in the target $ U_{ID^*}  $ value.  \Chall~sets $ bP $ as a part of the challenge ciphertext (i.e., $R^*\gets bP) $. 
	
	After the \AII~outputs a forgery,  \Chall~can extract the solution to the CDH problem since it has knowledge over the $ t $ secret values and  $ \beta^*$.  
	
\end{proof}
%This thorem is followed by a series of lemmas that will be proved in the appendix.  The method of the proof is similar ot those in \cite{BonehID-Based,PatersonCertLess}. We show that the IND-CCA security   of our \CLE  can be related to the IND-CCA security of a traditonal  public key encryption Though the lemmas we will show that the security of our \CLE scheme can be reduced 

%\begin{lemma}
%The proposed signature 
%	\end{lemma}

	\begin{Lemma}\label{lem:PKrep}
	A public key replacement attack by $ \A_I $  is not practical since it will falsify the private key. 
\end{Lemma} 

\begin{proof}
	Note that if $ \A_I  $ replaces $ U $ with a new value $ U' $ (which it might know the corresponding secret key), then the existing $ Q $ will be falsified since $ Q= U+W $ and this will also falsify the current partial private key component $ y  $ since it is computed based on the indexes that are obtained by computing $ \hash_1(ID,Q) $. Also note that, if $ \A_I  $ can obtain the original $ (\alpha, U) $, given $ Q $ is public, it can compute $ W $, however, $ W  $ is merely the commitment of $ \beta  $ and it does not disclose any information about $ \beta $.   The public key replacement attack in our security proof is possible if $ \A_I $ requests a new partial private key for each new $ U' $ .
\end{proof}

\begin{Lemma}\label{lem:IDSig}
	If an adversary  $ \A_I $ can break the EU-CMA security of the signature scheme proposed in  Algorithm \ref{alg:CLS} ,  then one can build another algorithm \Chall~ that runs $ \AI $ as subroutine and breaks a random instance of  the ECDLP ($ P,aP $).
\end{Lemma}
\begin{proof}
	Due to the space constraint, here we give the high level idea of our proof. 	We let \AI~be as   in Definition \ref{def:CertlessEUCMA}, then we can build another algorithm \Chall~that uses \AI~as a subroutine, and upon \AI's successful forgery, solves a random instance of the   ECDLP   ($ P,aP $).  \Chall~knows the $ t $ secret values $ v_i $'s and, similar to the proof of Theorem \ref{thm:IND-CLE-CCA}, it sets $ aP $ as a part of the target user's public key (i.e., $Q_{ID^*} \gets aP$).  Most of the simulation steps are like the ones in the proof of  Theorem \ref{thm:IND-CLE-CCA}. At the end of the simulation phase,  \AI~outputs a forgery signature $(s^*_1,e^*_1) $, the proof then uses the forking lemma \cite{Pointcheval1996} to run the adversary again to obtain a second forgery  $(s^*_2,e^*_2) $, using the same random tape. Our proof will follow the same approach as in \cite{DBLP:conf/africacrypt/GalindoG09} which is very similar to the proof  in  \cite{Schnorr91}. Given two forgeries and the knowledge of \Chall~on the $ v_i's $ and $ \alpha_{ID}^* $, \Chall~can compute $ a  $ and solve the ECDLP. Note that similar to Schnorr \cite{Schnorr91} the security of the scheme will be non-tight due to the forking lemma. 
\end{proof}

\noindent\textbf{Parameters Selection for ($ t,k $).} Parameters $(t,k)$  should be selected such that the probability $\frac{q_{\hash_1}\cdot k!}{2^{\gamma}}$ is negligible. Considering that $\gamma= k  \log_2{t}$ (since $k$ indexes that are $\log_2{t}$-bit long are selected with the hash output), this gives us $\frac{q_{\hash_1}\cdot k!}{2^{k \abs{t}}}$. We further elaborate on some choices of $(t,k)$ along with their performance implications in \autoref{sec:Performance}.

%Notice that the simulation of the $ $ \ $hash_i  $ for $ i\in\{1,2,3,4\} $  oracle  is identical to the one in the real scheme with a high probability (we negligent the in). 

%Forgery conditions on the fact that the adv should have asked $ \hash_3 $, or will in the future ask it. We rely on the fact that \Chall~knows $ v_i's $. 

%% file: perf.tex
\section{Performance Analysis and Comparison} \label{sec:Performance}
\newcommand{\mec}{\ensuremath {m_{EC}}{\xspace}}
\newcommand{\aec}{\ensuremath {a_{EC}}{\xspace}}
\newcommand{\dmec}{\ensuremath {dm_{EC}}{\xspace}}
\addtolength{\parskip}{-0.5mm}
\setlength{\textfloatsep}{0.1cm}

We first present the  analytical  and then experimental performance analysis and comparison of our schemes with their counterparts.  We focus on the \emph{online} operations (e.g., encryption, signing, key exchange)  for which both our IDB and CL schemes have the same algorithms, rather than one-time (offline) processes like setup and key generation. Since the online operations are identical in IDB  and CL systems in our case,  we refer to them as ``Our Schemes'' in the following tables/discussions.  We consider the cost of certificate verification for schemes in traditional PKI. We only consider the cost of verifying and communicating the cost of one certificate, which is highly conservative since in practice (i.e., X.509) there are at minimum  two certificates in a certificate  chain. This number could be as high as ten certificates in some scenarios.  

%For these operations, both   our identity-based and certificateless schemes have the same algorithms (i.e., the compatibility), therefore, we regard them both as ``our schemes'' in the following tables/discussions. 

% We then present our parameters selection followed by our evaluation metrics and experimental setup. We conclude with  a detailed experimental comparisons of our schemes and their counterparts.  

\begin{table*}[t] 	
			\caption{Analytical Comparison of Public Key Encryption Schemes}\label{tab:Analytical:Enc} 
	\centering
	
\small
	%	\resizebox{\textwidth}{!}{
%	\vspace{0mm}\resizebox{0.95\textwidth}{!}{%
		
		\begin{threeparttable}		\footnotesize
			\begin{tabular}{| c || c | c | c | c | c | c | c | c |}
				\hline
				\textbf{Scheme} & \makecell{\textbf{$ \sk $ } }& \textbf{Enc}  & \textbf{Comm.} & \textbf{Dec}  & \textbf{ $ \pk $ } &  \textbf{ $ mpk $ } & \makecell{\textbf{System}\\ \textbf{Type}}  & $ \kappa^\dagger $\\ \hline
				
				ECIES~\cite{ECIES_Shoup} & $\abs{q}$ & $2\,m_{EC}+\dmec$  & $\makecell{2\abs{p} +b+ \\d+CF}$ & $\mec$ & $\abs{p}$ &-&TD&128\\ \hline
				
				BF~\cite{BonehID-based}& $\abs{p}$ & $ bp+\mec+ex$ &  $\abs{p} +b+|M|$  & $bp +\mec$ & $\abs{p}$ &$ \abs{p} $&IB&80\\ \hline
				
				AP~\cite{PatersonCertLess} & $\abs{p}$ & $3\,bp+\mec+ex$ & $\abs{p} +b+|M|$ & $bp+ \mec$ & $2\abs{p}$ &$ \abs{p} $&CL&80\\ \hline

				BSS~\cite{BaekCertLessPairingFree} &  $2\,\abs{q}$ & $ 4\,ex+m $& $ \abs{p} +b+|M|$& $ 3\,ex $&$ 2|p| $ &$ \abs{p} $&CL&128\\ \hline 
				
				WSB~\cite{BertinoCertLessDrone} & $ \abs{p}+\abs{q}$ & $3\,\mec+2\,\aec $ & $2\abs{q}$ + $\abs{c}$ & $2 \,\mec$ & $2\abs{p}$ &$ \abs{p} $&CL&128\\ \hline \hline 
				
				%			Our ID-Based & & & & & \\ \hline
				
				Our Schemes & $\abs{p}$ & $2\,\mec + k \,\aec $ &  $ \abs{p} +b+|M|  $& $2\,\mec$ & $ \abs{p}$&$ t\,\abs{p} $&IB/CL  &128\\ \hline
				
			\end{tabular}
			\begin{tablenotes}[flushleft] \scriptsize{
					$ \dagger $ Denotes the security bit.  \textbf{Enc}, \textbf{Dec}, and \textbf{Comm.} represent  encryption, decryption, and communication load (bi-directional), respectively. $\mec$,  $\aec$,  and $\dmec$ denote the costs of EC scalar multiplication,  EC addition, and double scalar multiplication over modulus $p$, respectively.   $ m $,   $ ex $ and   $bp$ denote multiplication,   exponentiation   and   pairing operation, respectively. $k$ is the BPV parameter that shows how many precomputed pairs are selected in the online phase. $ b,d $ and $ CF $ denote block/key size for symmetric key encryption, message digest (i.e., MAC) size and size of the certificate, respectively. $ M $ denotes message space size. TD, IDB, and CL represent traditional public key cryptography,  identity-based cryptography, and certificateless cryptography, respectively. 
				}
			\end{tablenotes}

		\end{threeparttable}
%	}

\end{table*}
%\raggedbottom
\begin{table*}[t]
		
		\caption{Analytical Comparison of Digital Signature Schemes}\label{tab:Analytical:Signature}\vspace{0mm}
	\centering
 	\small
% \resizebox{0.95\textwidth}{!}{%
	%	\resizebox{0.95\textwidth}{!}{%
	%	\resizebox{\textwidth}{!}{
	\vspace{0mm}
	\begin{threeparttable}
		\begin{tabular}{| c || c | c | c | c | c | c | c | c |} 
			\hline
			\textbf{Scheme} &\makecell{\textbf{$ \sk $ } } & \textbf{\makecell{Sign}}  & \textbf{Comm.} & \textbf{\makecell{Verify}} &  \textbf{ $ \pk $ } & $ mpk $ & \makecell{\textbf{System}\\ \textbf{Type}} &$ \kappa $\\ \hline 
			
			Schnorr~\cite{Schnorr91} & $\abs{q}$ & $\mec$ & $2\abs{q} + CF$ & $2\dmec$ & $\abs{p}$ & -& TD& 128\\ \hline

			GG~\cite{DBLP:conf/africacrypt/GalindoG09}& $\abs{q}$ &  $\mec$ &  $2\abs{p} +|q|$  & $2\,\dmec$ & $\abs{p}$ &$ \abs{p}  $ &IDB&128\\ \hline
			AP~\cite{PatersonCertLess} & $\abs{p}$ & $2\,\mec+\aec+bp$ & $\abs{q}+\abs{p}$ & $4\,bp+ex$& $2\abs{p}$& $ \abs{p}  $&CL& 80\\ \hline 
			
			%				ZWXF~\cite{CerteLessSignatureZWXF} & $2\abs{q}$ & $3$ $Emul$ & $2\abs{q}$ & $4$ $P$  & $\abs{q}$ &CL \\ \hline   
			
			KIB~\cite{CerteLessSignatureKIB} & $\abs{q}$ & $\mec$ & $\abs{q}+\abs{p}$  & $3\,\mec$ & $3\,\abs{p}$ & $ 2\,\abs{p} $& CL &128\\ \hline  \hline
			
			%			AQ-Hang~\cite{} & & & & & \\ \hline \hline 
			
			%			Our ID-Based & & & & & \\ \hline
			
			Our Schemes& $\abs{q}$ & $\mec$ & $2\abs{q}$ & $\dmec+ k \;\aec$ & $\abs{p}$ & $ t\,\abs{p} $&IDB/CL& 128\\  \hline
			
		\end{tabular}
		%		\begin{tablenotes}[flushleft] \scriptsize{
		%				$Emul$ and $Eadd$ denote the costs of EC scalar multiplication over modulus $q$ and EC addition over modulus $q$, respectively. $H$ and $Mulq$ denote a cryptographic hash and a modular multiplication over modulus $q$, respectively. $k$ is the BPV parameter that shows how many precomputed pairs are selected in the online phase. Suggested value for $k=18$ \cite{}. $r'$ is the constant number of public keys (servers) that the node will communicate.  \\
		%				$\Gamma = n\cdot \abs{q}$ where $n$ is the number of precomputed pairs. Parameter sizes for $n$, $q$ and $H$ are: $n = 160$, $\abs{q} = 192$ bit, $\abs{H} = 256$ bit
		%			}
		%		\end{tablenotes}
	\end{threeparttable}
% 	}

\end{table*}\raggedbottom  
\begin{table*}[t]

		\caption{Analytical Comparison of Key Exchange Schemes}\label{tab:Analytical:KeyExchange}
	\small
	\centering
	%	\small
	
	%	\resizebox{\textwidth}{!}{

%	\resizebox{0.95\textwidth}{!}{%
		\begin{threeparttable}
			\begin{tabular}{| c || c | c |  c | c | c| c | c |}
				\hline
				\textbf{Scheme} &  \makecell{\textbf{$ \sk $ } }  & \textbf{User Comp.}  & \textbf{Comm.} & \textbf{ $ \pk $ } &$ mpk $  &\makecell{\textbf{System}\\ \textbf{Type}}\\ \hline 
				
				Ephemeral ECDH~\cite{DH} & $\abs{q}$ & $ 2 \mec $   & $2\abs{{p}} +CF$ &  $\abs{p}$ & -&TD\\ \hline
				
				ECHMQV~\cite{DBLP:conf/crypto/Krawczyk05} & $\abs{q}$ & $3\,\mec $ & $ 2\abs{p} + CF$ &  $\abs{p}$ & -&TD \\ \hline
				
				TFNS~\cite{DBLP:conf/esorics/TomidaFNS19} & $ \abs{p}$ & $bp + 5\,\mec$ & $ \abs{p}$ & $ \abs{p}$ & $ \abs{p} $&IDB\\ \hline  
				
				AP~\cite{PatersonCertLess} & $\abs{p}$ & $4\,bp+ex$ & $3\abs{p}$ & $2\abs{p}$ &$ \abs{p} $ &CL \\ \hline
				
				YT~\cite{CertLessYang} & $2\abs{q}+ \abs{p}+|s|^\dagger$ & $ 3\,\dmec +5\,\mec  $ & $3\abs{p} + \abs{s}$ & $2\abs{p} + \abs{s}$ & $ 2\,\abs{p} $  &CL \\ \hline \hline
				
				%				AQ-Hang~\cite{} & $\abs{q}$ & $4$ $Emul$ + $2$ $Eadd$  & $2\abs{q}$& $\abs{q}$  \\ \hline

				Our Scheme  & $\abs{q}$ & $3\,\mec + (k+1)\,\aec$ & $2\abs{p}$ &  $\abs{p}$&  $ t\,\abs{p} $&IDB/CL\\ \hline

			\end{tabular} 
			User Comp. denotes user computation. \\
			$ \dagger \abs{s} $ denotes the output of a signature scheme that the authors in \cite{CertLessYang} use in their scheme. 
		\end{threeparttable}
%	}

\end{table*}  \vspace{0mm}
\noindent\textbf{Analytical Performance Analysis and Comparison}  We present a detailed analytical performance comparison of our schemes with their counterparts for public key encryption/decryption, digital signature and key exchange in  \autoref{tab:Analytical:Enc}, \autoref{tab:Analytical:Signature} and \autoref{tab:Analytical:KeyExchange}, respectively. 

Our schemes have significantly lower communication overhead than their PKI-based counterpart in all cryptosystems as they do not require  the transmission of   certificates. As discussed above and also elaborated in Section \ref{subsec:ExperimentalComp}, this translates into substantial bandwidth gain as well as computational efficiency since the certification verification overhead is also lifted. Moreover, in almost all instances, our schemes also offer a lower end-to-end computational overhead compared to their PKI-based counterparts. Our schemes also offer a lower end-to-end computational delay than that of  all of their IDB and CL counterparts in all cryptosystems, with generally equal private and public key sizes. However, the master public key size of our scheme is larger than all of their counterparts.

\begin{table*}[t]
 
		\caption{Public Key Encryption Schemes on Commodity Hardware}	\label{tab:Laptop:Enc}\vspace{0mm}
	\centering
	\small
	
	%	\resizebox{\textwidth}{!}{

	\begin{threeparttable}
		\begin{tabular}{| c || c | c | c | c | c | c | c |}
			\hline
			\textbf{Scheme} & \makecell{\textbf{$ \sk $ } }& \textbf{Enc}  & \textbf{Comm.} & \textbf{Dec}  & \textbf{ $ \pk $ }  & \textbf{ $ mpk $ }  &  \textbf{E2E Delay}\\   \hline
			
			ECIES~\cite{ECIES_Shoup} & $32$ & $55$ & $690^\dagger   + \abs{M}$ & $21$ & $32$ &-& $76$\\ \hline
			
			BF~\cite{BonehID-based} & $32$ & $\approx 2000$ & $48 + \abs{M}$ & $\approx 2000$ & $32$ &32&  $\approx 4000$ \\ \hline
			
			AP~\cite{PatersonCertLess} & $32$ & $\approx 6000$ & $48 + \abs{M}$ & $\approx 2000$ & $64$ &32& $\approx 8000$ \\ \hline
			
			%			AQ-Hang~\cite{} & & & & & \\ \hline 
			
			BSS~\cite{BaekCertLessPairingFree} & $64$ & $73 $ & $64 + \abs{M}$ & $ 53 $ & $ 64 $ &32& $126$\\ \hline 
			
			WSB~\cite{BertinoCertLessDrone} & $64$ & $53$  & $64$ + $\abs{M}$ & $41$ & $64$ &32& $94$ \\ \hline \hline 
			
			%			Our ID-Based & & & & & \\ \hline
			
			Our Schemes & $32$ & $39$ & $48$ + $\abs{M}$ & $33$ & $32$ &32K& $72$ \\ \hline
			
		\end{tabular}
		\begin{tablenotes}[flushleft] \scriptsize{
				All sizes are in Bytes, and all computations are in microseconds. \\
				$\dagger$ We assume the certificate size is 578 Bytes, the   size is given in RFC 5280 \cite{rfc5280}. 
				%		$\ddagger$ The signature size is considered as 64 bytes. 
			}
		\end{tablenotes}
	\end{threeparttable}
	%	}

\end{table*}
% ================================ CLS Laptop  

\begin{table*}[t]
		%	\resizebox{\textwidth}{!}{
 
		\caption{Digital Signature Schemes on Commodity Hardware}	\label{tab:Laptop:Signature} \vspace{0mm}
	\centering
	\small

	\begin{threeparttable}
		\begin{tabular}{| c || c | c | c | c | c | c | c |}
			\hline
			\textbf{Scheme} &\makecell{\textbf{$ \sk $ } } & \textbf{\makecell{Sign}}  & \textbf{Comm.} & \textbf{\makecell{Verify}} &  \textbf{ $ \pk $ } &  { $ mpk $ } &  \textbf{E2E Delay} \\ \hline

			Schnorr~\cite{Schnorr91} & $32$ & $12$ & $642$ & $44$ & $32$ &-& $56$ \\ \hline
			
			GG~\cite{DBLP:conf/africacrypt/GalindoG09} & $32$ &  $12$ &$96$  & $44$  &  $32$ &32& $56$\\ \hline
			
			AP~\cite{PatersonCertLess} & $32$ & $\approx 2000$ & $64$ & $\approx 8000$ & $64$ &32& $\approx 10000$ \\ \hline 
			
			%			ZWXF~\cite{CerteLessSignatureZWXF} & $64$ & $52$ & $64$ & $\approx 8000$  & $32$ &  $\approx 8000$\\ \hline   
			
			KIB~\cite{CerteLessSignatureKIB} & $32$ & $20$ & $64$ & $61$ & $96$ &64& $81$ \\ \hline  \hline
			
			%			AQ-Hang~\cite{} & & & & & \\ \hline \hline 
			
			%			Our ID-Based & & & & & \\ \hline
			
			Our Schemes & $32$ & $12$ & $64$ & $27$ & $32$ &32K& $39$ \\ \hline
			
		\end{tabular}
		\begin{tablenotes}[flushleft] \scriptsize{
				All sizes are in Bytes, and all computations are in microseconds. \\
			}
		\end{tablenotes}
	\end{threeparttable}
	%	}

\end{table*}
\begin{table*}[h!]

	\caption{Key Exchange Schemes on Commodity Hardware}	\label{tab:Laptop:KeyExchange} 
	\centering
		\small
	
	%	\resizebox{\textwidth}{!}{

	\begin{threeparttable}
		\begin{tabular}{| c || c | c |  c | c | c | c | }
			\hline
			\textbf{Scheme} &  \makecell{\textbf{$ \sk $ } }  & \textbf{User Comp.}  & \textbf{Comm.} & \textbf{ $ \pk $ } & $ mpk $ &  {\textbf{E2E Delay} } \\ \hline
			
			Ephemeral ECDH~\cite{DH} & $32$ & $55$ & $642$ &  $32$ &-& $110$\\ \hline
			
			ECHMQV~\cite{DBLP:conf/crypto/Krawczyk05} & $32$ & $74$ & $642$ & $32$ & -&$148$ \\ \hline
			
			TFNS~\cite{DBLP:conf/esorics/TomidaFNS19} & $32$ & $\approx 2000$ & $32$ & $32$ &32& $\approx 4000$ \\ \hline  
			
			AP~\cite{PatersonCertLess} & $32$ & $\approx 8000$ & $96$ & $64$ &32& $\approx 16000$ \\ \hline
			
			YT\textsuperscript{$\dagger$}~\cite{CertLessYang} & $160$ & $157$ & $160$ & $128$ & 64&$314$ \\ \hline \hline
			
			%				AQ-Hang~\cite{} & $\abs{q}$ & $4$ $Emul$ + $2$ $Eadd$  & $2\abs{q}$& $\abs{q}$  \\ \hline
			
			%				WSB~\cite{BertinoCertLessDrone} & $2\abs{q}$ & $4$ $Emul$ + $2$ $Eadd$ & & $2\abs{q}$ \\ \hline \hline
			
			Our Scheme & $32$ & $57$ & $64$ &  $32$& 32K&$114$\\ \hline
			
			%			Our ID-Based & $32$ & $57$ & $64$ &  $32$\\ \hline
			
			%			Our Certificateless & $32$ & $69$ & $64$ & $32$ \\ \hline
			
		\end{tabular}
		\begin{tablenotes}[flushleft] \scriptsize{
				All sizes are in Bytes, and all computations are in microseconds. \\
				$\dagger$ The signature size is considered as 64 bytes. 
			}
		\end{tablenotes}
	\end{threeparttable}

\end{table*}  

\label{subsec:ExperimentalComp}\noindent\textbf{Experimental Performance Analysis and Comparison:} 
We now further elaborate on  the details of our performance analysis and comparison with experimental results. We conduct experiments on both commodity hardware and low-end embedded devices that are typically found in IoT systems to objectively assess the performance of our schemes as well as their counterparts. Our open-sourced implementation is available via the following link. 

\begin{center}
\MYhref{https://github.com/Rbehnia/CertFreeSystems}{https://github.com/Rbehnia/CertFreeSystems}
\end{center}

\noindent $\bullet$~{\em Experiments on Commodity Hardware:} We used an i7 Skylake laptop equipped with a 2.6 GHz CPU and 12 GB RAM in our experiments. We implemented our schemes on the FourQ curve~\cite{FourQ} which offers fast elliptic curve operations for $\kappa=128$-bit security. We instantiated our random oracles with blake2 hash function\footnote{\url{http://131002.net/blake/blake.pdf}}, which offers high efficiency and security. For our parameters, we selected $ k=18 $ and $ t=1024 $. We conservatively estimated the costs of our counterparts based on the microbenchmarks on our evaluation setup of (i) FourQ curve for schemes that do not require pairing and (ii) PBC library\footnote{\url{https://crypto.stanford.edu/pbc/}} on a curve with $\kappa=80$-bit security (we used the most efficient alternative for them) for schemes that require pairing. %We estimated their costs on fast FourQ curve to be fair with them. 

%\noindent \textbf{Hardware Configuration and Software Libraries.}  

As depicted in \autoref{tab:Laptop:Enc}, the encryption and decryption algorithms of our schemes are more efficient than their counterparts in the identity-based and certificateless settings. More specifically, the end-to-end delay of our schemes is  $ \approx25$\%  lower than that in \cite{BertinoCertLessDrone}, which is specifically suitable for aerial drones.  One could also notice how the communication overhead is lower in certificateless and identity-based schemes since there is no need for certificate transmission.
\begin{table*}[t]

		\caption{Public Key Encryption Schemes on 8-bit AVR processor} \label{tab:8bit:Enc}
	\centering
		\small
	
	%	\resizebox{\textwidth}{!}{

	\begin{threeparttable}
		\begin{tabular}{| c || c | c | c | c | c | c |}
			\hline
			\textbf{Scheme} & \makecell{\textbf{$ \sk $ } }& \textbf{Enc}  & \textbf{Comm.} & \textbf{Dec}  & \textbf{ $ \pk $ }  & $ mpk $\\   \hline 
			
			ECIES~\cite{ECIES_Shoup} & $32$ & 17 950 967 & $610 + \abs{c}$ & 6 875 861 & $32$ &-\\ \hline
			
			BF~\cite{BonehID-based} & $32$ & 60 802 555 & $48 + \abs{c}$ & 56 278 012 &$32$&32\\ \hline
			
			AP~\cite{PatersonCertLess} & $32$ & 166 213 087 & $48 + \abs{c}$ & 58 912 609 & $64$ &32\\ \hline
			
			%			AQ-Hang~\cite{} & & & & & \\ \hline 
			
			BSS~\cite{BaekCertLessPairingFree} & $64$ & 22 791 835 & $64 + \abs{M}$ & 16 590 321 & $ 64 $ &32\\ \hline 
			
			WSB~\cite{BertinoCertLessDrone} & $64$ & 17 091 636 & $64$ + $\abs{c}$ & 13 631 755 & $64$& 32\\ \hline \hline 
			
			%			Our ID-Based & & & & & \\ \hline
			
			Our Schemes  & $32$ & 11 789 632 & $48$ + $\abs{c}$ & 9 883 161 & $32$&32K\\ \hline
			
		\end{tabular}
		\begin{tablenotes}[flushleft] \scriptsize{
				All sizes are in Bytes, and all computations are in CPU Cycles. 
			}
		\end{tablenotes}
	\end{threeparttable}

	%	}
\end{table*}

\begin{table*}[t]

		\caption{Digital Signature Schemes on 8-bit AVR processor} \label{tab:8bit:Signature}
	\centering
	\small
	
	%	\resizebox{\textwidth}{!}{
 
	\begin{threeparttable}
		\begin{tabular}{| c || c | c | c | c | c | c |}
			\hline
			\textbf{Scheme} &\makecell{\textbf{$ \sk $ } } & \textbf{\makecell{Sign}}  & \textbf{Comm.} & \textbf{\makecell{Verify}} &  \textbf{ $ \pk $ } & $ mpk $  \\ \hline
			
			Schnorr~\cite{Schnorr91} & $32$ & 4 263 298 & $642$ & 17 902 958 & $32$ & -\\ \hline
			
			GG~\cite{DBLP:conf/africacrypt/GalindoG09} & $32$ & 4 263 298 & $96$  &  17 902 958 &   $32$ & 32\\ \hline
			
			AP~\cite{PatersonCertLess} & $32$ & 62 487 032 & $64$ & 221 226 015 & $64$ & 32\\ \hline 
			
			%			ZWXF~\cite{CerteLessSignatureZWXF} & $64$ & 16 981 570 & $64$ & 220 265 370 & $32$ \\ \hline   
			
			KIB~\cite{CerteLessSignatureKIB} & $32$ & 7 025 861 & $64$ & 20 617 583 & $96$ &64 \\ \hline  \hline
			
			%			AQ-Hang~\cite{} & & & & & \\ \hline \hline 
			
			%			Our ID-Based & & & & & \\ \hline
			
			Our Schemes  & $32$ & 4 263 298 & $64$ & 10 955 369 & $32$&32K  \\ \hline
			
		\end{tabular}
		\begin{tablenotes}[flushleft] \scriptsize{
				All sizes are in Bytes, and all computations are in CPU Cycles. 
			}
		\end{tablenotes}
	\end{threeparttable}
	%	}

\end{table*}

\begin{table*}[t]

		\caption{Key Exchange Schemes on 8-bit AVR processor} \label{tab:8bit:KeyExchange} 
	\centering
	\small
	
	%	\resizebox{\textwidth}{!}{

	\begin{threeparttable}
		\begin{tabular}{| c || c | c |  c | c | c |}
			\hline
			\textbf{Scheme} &  \makecell{\textbf{$ \sk $ } }  & \textbf{User Comp.}  & \textbf{Comm.} & \textbf{ $ \pk $ } & $ mpk $\\ \hline
			
			Ephemeral ECDH~\cite{DH} & $32$ & 18 039 710& $642$ &  $32$ &-\\ \hline
			
			ECHMQV~\cite{DBLP:conf/crypto/Krawczyk05} & $32$ & 24 601 857 & $642$ & $32$& -\\ \hline
			
			TFNS~\cite{DBLP:conf/esorics/TomidaFNS19} & $32$ & 82 356 781 & $32$ & $32$ &32 \\ \hline  
			
			AP~\cite{PatersonCertLess} & $32$ & 221 226 015 & $96$ & $64$  &32 \\ \hline
			
			YT~\cite{CertLessYang} & $160$ & 54 212 824 & $160$ & $128$ &64 \\ \hline \hline
			
			%				AQ-Hang~\cite{} & $\abs{q}$ & $4$ $Emul$ + $2$ $Eadd$  & $2\abs{q}$& $\abs{q}$  \\ \hline
			
			%				WSB~\cite{BertinoCertLessDrone} & $2\abs{q}$ & $4$ $Emul$ + $2$ $Eadd$ & & $2\abs{q}$ \\ \hline \hline 
			
			Our Scheme & $32$ & 18 015 493 & $64$ &  $32$  & 32K\\ \hline
			
			%			Our ID-Based & $32$ & 18 015 493 & $64$ &  $32$\\ \hline
			
			%			Our Certificateless & $32$ & 20 995 341 & $64$ & $32$ \\ \hline
			
		\end{tabular}
		\begin{tablenotes}[flushleft] \scriptsize{ 
				All sizes are in Bytes, and all computations are in CPU Cycles. 
			}
		\end{tablenotes}
	\end{threeparttable}
	%	}

\end{table*}

As shown in \autoref{tab:Laptop:Signature}, our schemes enjoy from the fastest verification algorithms among all its counterparts. This is again due to the novel way the user keys are derived and results in $ 30 $\% and $ 52 $\% faster end-to-end delay as compared to its most efficient identity-based  \cite{DBLP:conf/africacrypt/GalindoG09} and certificateless \cite{CerteLessSignatureKIB} counterparts, respectively. \emph{One may notice that although  schemes in \cite{DBLP:conf/africacrypt/GalindoG09,Schnorr91,CerteLessSignatureKIB}, along with our schemes, all  require a scalar multiplication in their signature generation (see Table~\ref{tab:Analytical:Signature}), their experimental costs differ. The reason for this discrepancy is the fact that the cost of scalar multiplication over the generator  $P$  is faster than the scalar multiplication over any curve points, and these differences are considered in the experimental evaluations.}
%\RB{@ozgur, please add that line about diff multiplications here}

 As shown in  \autoref{tab:Laptop:KeyExchange}, our schemes' performance is similar to that in their counterparts in traditional PKI setting \cite{DH,DBLP:conf/crypto/Krawczyk05}. However, they outperform the most efficient counterpart in certificateless setting \cite{CertLessYang} by having $ 65 $\% lower end-to-end delay and $ 60 $\% smaller load for communication.

\noindent $\bullet$ {\em Experiments on Low-End Device:}  We used an 8-bit AVR ATmega 2560 microprocessor to evaluate the costs of our schemes on an IoT device. AVR ATmega 2560 is a low-power microprocessor with 256 KB flash, 8 KB SRAM, 4 KB EEPROM, and operates at 16 MHz frequency. %MAYBE ADD SUCH A SENTENCE: It is extensively used in practice, medical devices drones etc. (Cite IoD-Crypt?).
%We again used the FourQ curve~\cite{FourQ} and blake2 hash function~\cite{blakeHash} in our implementations. 
We used the 8-bit AVR library of the FourQ curve presented in~\cite{FourQ8bit}. For our counterparts, we again conservatively estimated their costs based on microbenchmarks in (i) FourQ curve 8-bit AVR implementation~\cite{FourQ8bit}, and (ii) NanoECC~\cite{NanoECC_Pairing}, that implements a curve that supports pairings on 8-bit AVR microprocessors and offers $\kappa=80$-bit security.

 As depicted  in \autoref{tab:8bit:Enc}, our schemes outperform all of their identity-based and certificateless counterparts  and have a more efficient encryption algorithm than \cite{ECIES_Shoup}. Our decryption algorithms, while   being more efficient than all of their   identity-based and certificateless counterparts, are slightly less efficient than the one in \cite{ECIES_Shoup}. Similar to the trend  in the analytical performance, our signature schemes   outperform their counterparts. As \autoref{tab:8bit:Signature} shows, our schemes' signing algorithm are amongst the most efficient ones, while the verification algorithm outperforms all the counterparts with similar communication overhead.

\noindent\textbf{Limitations:} The main limitation of  our schemes is the size of the master public key. Note that if there are different TTP in different domains and users often communicate with the users in those domains, it would make sense to store different $ mpk $. Otherwise, the users only need to store   $ mpk $ for their own systems. We can reduce the size of the $ mpk $ in exchange for a small performance loss. For instance, with $k =32  $, we can reduce the size of the $ mpk $ by four times. 

%
%
%
%
%Cite in related work: Hiearchical IBE, IBS~\cite{Hierarchical_IBE_IBS}, Leakage resilient certificateless public key encryption~\cite{Leakage_Resilient_Certificateless_PKE}, Continuous leakage resilient IBE~\cite{Leakage_Resilient_IBE}

%% file: RelatedWork.tex
\section{Related Work}
There is a comprehensive literature covering different aspects of IDB and CL systems. Remark that   most of the closely related works have been discussed in Section 3 and   5 in terms of security models and performance metrics. Overall, the main difference in our work is to focus on the achievement of inter-compatibility between IDB and CL with a high efficiency, with respect to existing alternatives.

The idea of IDB cryptography was proposed by Shamir \cite{ShamirID-Based}. However, the first practical instance of such schemes was proposed later by Boneh and Franklin \cite{BonehID-based} using bilinear pairing.  %Sakai and Kasahara  \cite{Sakai2003} proposed  an encryption scheme based on the exponent inversion, where encryption does not require any pairing computation and there is only one pairing computation in the decryption algorithm.    
To get the full adaptive-identity, chosen-ciphertext  security guarantees without sacrificing performance, Boyen  \cite{DBLP:journals/ijact/Boyen08}   described  an  augmented versions of  the scheme in \cite{IdendityBasedEncSelectiveeIDBonehB04a}   in the random-oracle model. However, the augmented version also requires multiple pairing computations in the decryption algorithm. Following \cite{BellareIDBased}, several pairing-free signature scheme were proposed. Galindo and Garcia \cite{DBLP:conf/africacrypt/GalindoG09}  proposed a lightweight IDB signature scheme   based on  \cite{Schnorr91} with reduction to the discrete logarithm problem. 

CL cryptography \cite{PatersonCertLess}  was proposed to address the private key escrow problem in IDB systems. In the same paper, the authors proposed an IND-CCA encryption scheme along with a signature and key exchange schemes.  Following their work, Baek et al. \cite{BaekCertLessPairingFree} proposed the first IND-CCA  secure certificateless encryption scheme without pairing.  The scheme is constructed using Schnorr-like signatures in partial private key generation algorithm. Recently Won et al. \cite{BertinoCertLessDrone} proposed another efficient IND-CCA encryption scheme that is specifically used for key encapsulation mechanisms.  There has been a number of works that focus on the security models  of certificateless systems. In most  of the proposed models (e.g., \cite{PatersonCertLess}) a \TII adversary is assumed to generate the keys honestly, and initiate the attacks only after the setup phase. \\
%However, Au et al. \cite{DBLP:journals/iacr/AuCLMWY06}  defined a security model against a \TII adversary where it is assumed that the adversary is able to run the setup algorithm maliciously. In our schemes, we follow the security model in \cite{BaekCertLessPairingFree} which  can be extended to address the attacks in the model proposed in \cite{DBLP:journals/iacr/AuCLMWY06}. 

\noindent \textbf{Acknowledgments.} This work is supported by the NSF CAREER Award CNS-1917627, an unrestricted gift via Cisco Research Award and the Department of Energy Award DE-OE000078

%\noindent \textbf{Acknowledgment. }This work is supported by NSF   award  \#1652389.

% initial proposals relied on bilinear pairing operations which is known to be quite costly.  Following the work in \cite{BellareIDBased}, several works were proposed 

%% file: Main.bbl
% Generated by IEEEtranS.bst, version: 1.12 (2007/01/11)
\begin{thebibliography}{10}
\providecommand{\url}[1]{#1}
\csname url@samestyle\endcsname
\providecommand{\newblock}{\relax}
\providecommand{\bibinfo}[2]{#2}
\providecommand{\BIBentrySTDinterwordspacing}{\spaceskip=0pt\relax}
\providecommand{\BIBentryALTinterwordstretchfactor}{4}
\providecommand{\BIBentryALTinterwordspacing}{\spaceskip=\fontdimen2\font plus
\BIBentryALTinterwordstretchfactor\fontdimen3\font minus
  \fontdimen4\font\relax}
\providecommand{\BIBforeignlanguage}[2]{{%
\expandafter\ifx\csname l@#1\endcsname\relax
\typeout{** WARNING: IEEEtranS.bst: No hyphenation pattern has been}%
\typeout{** loaded for the language `#1'. Using the pattern for}%
\typeout{** the default language instead.}%
\else
\language=\csname l@#1\endcsname
\fi
#2}}
\providecommand{\BIBdecl}{\relax}
\BIBdecl

\bibitem{PatersonCertLess}
S.~S. Al-Riyami and K.~G. Paterson, ``Certificateless public key
  cryptography,'' in \emph{Advances in Cryptology - ASIACRYPT 2003}, C.-S.
  Laih, Ed.\hskip 1em plus 0.5em minus 0.4em\relax Springer Berlin Heidelberg,
  2003, pp. 452--473.

\bibitem{ManHo07}
M.~H. Au, Y.~Mu, J.~Chen, D.~S. Wong, J.~K. Liu, and G.~Yang, ``Malicious kgc
  attacks in certificateless cryptography,'' in \emph{2nd ACM Symposium on
  Information, Computer and Communications Security}, ser. ASIACCS, 2007, pp.
  302--311.

\bibitem{BaekCertLessPairingFree}
J.~Baek, R.~Safavi-Naini, and W.~Susilo, ``Certificateless public key
  encryption without pairing,'' in \emph{Information Security}, J.~Zhou,
  J.~Lopez, R.~H. Deng, and F.~Bao, Eds.\hskip 1em plus 0.5em minus 0.4em\relax
  Springer Berlin Heidelberg, 2005, pp. 134--148.

\bibitem{ARIS}
R.~Behnia, M.~O. Ozmen, and A.~A. Yavuz, ``{ARIS}: Authentication for real-time
  {IoT} systems,'' in \emph{IEEE International Conference on Communications
  (ICC)}, ser. ICC.\hskip 1em plus 0.5em minus 0.4em\relax New York, NY, USA:
  ACM, 2019, pp. 1855--1867.

\bibitem{Tachyon}
R.~Behnia, M.~O. Ozmen, A.~A. Yavuz, and M.~Rosulek, ``Tachyon: Fast signatures
  from compact knapsack,'' in \emph{Proceedings of the 2018 ACM SIGSAC
  Conference on Computer and Communications Security}, ser. CCS '18.\hskip 1em
  plus 0.5em minus 0.4em\relax New York, NY, USA: ACM, 2018, pp. 1855--1867.

\bibitem{RandomOracleModel93}
M.~Bellare and P.~Rogaway, ``Random oracles are practical: A paradigm for
  designing efficient protocols,'' in \emph{Proceedings of the 1st ACM
  conference on Computer and Communications Security ({CCS} '93)}.\hskip 1em
  plus 0.5em minus 0.4em\relax NY, USA: ACM, 1993, pp. 62--73.

\bibitem{BellareIDBased}
M.~Bellare, C.~Namprempre, and G.~Neven, ``Security proofs for identity-based
  identification and signature schemes,'' in \emph{Advances in Cryptology -
  EUROCRYPT 2004}, C.~Cachin and J.~L. Camenisch, Eds.\hskip 1em plus 0.5em
  minus 0.4em\relax Berlin, Heidelberg: Springer Berlin Heidelberg, 2004, pp.
  268--286.

\bibitem{IdendityBasedEncSelectiveeIDBonehB04a}
D.~Boneh and X.~Boyen, ``Efficient {Selective-ID} secure identity-based
  encryption without random oracles,'' in \emph{Theory and Applications of
  Cryptographic Techniques ({EUROCRYPT} '04)}, 2004, pp. 223--238.

\bibitem{BonehID-based}
D.~Boneh and M.~K. Franklin, ``Identity-based encryption from the weil
  pairing,'' in \emph{Advances in Cryptology - {CRYPTO} 2001}, 2001, pp.
  213--229.

\bibitem{DBLP:journals/ijact/Boyen08}
X.~Boyen, ``A tapestry of identity-based encryption: practical frameworks
  compared,'' \emph{{IJACT}}, vol.~1, no.~1, pp. 3--21, 2008.

\bibitem{rfc5280}
D.~Cooper, S.~Santesson, S.~Farrell, S.~Boeyen, R.~Housley, and W.~Polk,
  ``Internet x. 509 public key infrastructure certificate and certificate
  revocation list (crl) profile,'' Internet Requests for Comments, RFC Editor,
  RFC 5280, May 2008.

\bibitem{FourQ}
C.~Costello and P.~Longa, ``Four${Q}$: Four-dimensional decompositions on a
  ${Q}$-curve over the mersenne prime,'' in \emph{Advances in Cryptology --
  ASIACRYPT 2015}, T.~Iwata and J.~H. Cheon, Eds.\hskip 1em plus 0.5em minus
  0.4em\relax Springer Berlin Heidelberg, 2015, pp. 214--235.

\bibitem{DH}
W.~Diffie and M.~Hellman, ``New directions in cryptography,'' \emph{{IEEE}
  Transactions on Information Theory}, vol. IT-22, pp. 644--654, November 1976.

\bibitem{DBLP:conf/crypto/FujisakiO99}
E.~Fujisaki and T.~Okamoto, ``Secure integration of asymmetric and symmetric
  encryption schemes,'' in \emph{Advances in Cryptology - {CRYPTO} '99}, 1999,
  pp. 537--554.

\bibitem{DBLP:conf/africacrypt/GalindoG09}
D.~Galindo and F.~D. Garcia, ``A schnorr-like lightweight identity-based
  signature scheme,'' in \emph{Progress in Cryptology - {AFRICACRYPT} 2009},
  2009, pp. 135--148.

\bibitem{SelfCertGIRAULTOriginal91}
M.~Girault, ``Self-certified public keys,'' in \emph{Proceedings of the 14th
  International Conference on the Theory and Application of Cryptographic
  Techniques (EUROCRYPT '91)}, 1991, pp. 490--497.

\bibitem{CertLessRevisited}
X.~Huang, Y.~Mu, W.~Susilo, D.~S. Wong, and W.~Wu, ``Certificateless signature
  revisited,'' in \emph{Information Security and Privacy}, J.~Pieprzyk,
  H.~Ghodosi, and E.~Dawson, Eds.\hskip 1em plus 0.5em minus 0.4em\relax
  Springer Berlin Heidelberg, 2007, pp. 308--322.

\bibitem{CerteLessSignatureKIB}
A.~Karati, S.~H. Islam, and G.~Biswas, ``A pairing-free and provably secure
  certificateless signature scheme,'' \emph{Information Sciences}, vol. 450,
  pp. 378 -- 391, 2018.

\bibitem{DBLP:conf/crypto/Krawczyk05}
H.~Krawczyk, ``{HMQV:} {A} high-performance secure diffie-hellman protocol,''
  in \emph{Advances in Cryptology - {CRYPTO} 2005}, 2005, pp. 546--566.

\bibitem{FourQ8bit}
Z.~Liu, P.~Longa, G.~C. C.~F. Pereira, O.~Reparaz, and H.~Seo, ``Four${Q}$ on
  embedded devices with strong countermeasures against side-channel attacks,''
  in \emph{Cryptographic Hardware and Embedded Systems -- CHES 2017},
  W.~Fischer and N.~Homma, Eds.\hskip 1em plus 0.5em minus 0.4em\relax Cham:
  Springer International Publishing, 2017, pp. 665--686.

\bibitem{DBLP:journals/corr/abs-1904-06829}
\BIBentryALTinterwordspacing
M.~O. Ozmen, R.~Behnia, and A.~A. Yavuz, ``Iod-crypt: {A} lightweight
  cryptographic framework for internet of drones,'' \emph{CoRR}, vol.
  abs/1904.06829, 2019. [Online]. Available:
  \url{http://arxiv.org/abs/1904.06829}
\BIBentrySTDinterwordspacing

\bibitem{Pointcheval1996}
D.~Pointcheval and J.~Stern, ``Security proofs for signature schemes,'' in
  \emph{Advances in Cryptology --- EUROCRYPT '96}, U.~Maurer, Ed.\hskip 1em
  plus 0.5em minus 0.4em\relax Springer Berlin Heidelberg, 1996, pp. 387--398.

\bibitem{HORS_BetterthanBiBa02}
L.~Reyzin and N.~Reyzin, ``Better than {BiBa}: Short one-time signatures with
  fast signing and verifying,'' in \emph{Proceedings of the 7th Australian
  Conference on Information Security and Privacy ({ACIPS '02})}.\hskip 1em plus
  0.5em minus 0.4em\relax Springer-Verlag, 2002, pp. 144--153.

\bibitem{Schnorr91}
C.~Schnorr, ``Efficient signature generation by smart cards,'' \emph{Journal of
  Cryptology}, vol.~4, no.~3, pp. 161--174, 1991.

\bibitem{ShamirID-Based}
A.~Shamir, ``Identity-based cryptosystems and signature schemes,'' in
  \emph{Advances in Cryptology}.\hskip 1em plus 0.5em minus 0.4em\relax
  Springer Berlin Heidelberg, 1985, pp. 47--53.

\bibitem{ECIES_Shoup}
V.~Shoup, ``A proposal for an iso standard for public key encryption,''
  Cryptology ePrint Archive, Report 2001/112, 2001,
  \url{https://eprint.iacr.org/2001/112}.

\bibitem{NanoECC_Pairing}
P.~Szczechowiak, L.~B. Oliveira, M.~Scott, M.~Collier, and R.~Dahab, ``Nanoecc:
  Testing the limits of elliptic curve cryptography in sensor networks,'' in
  \emph{Wireless Sensor Networks}, R.~Verdone, Ed.\hskip 1em plus 0.5em minus
  0.4em\relax Springer Berlin Heidelberg, 2008, pp. 305--320.

\bibitem{DBLP:conf/esorics/TomidaFNS19}
J.~Tomida, A.~Fujioka, A.~Nagai, and K.~Suzuki, ``Strongly secure
  identity-based key exchange with single pairing operation,'' in
  \emph{Computer Security - {ESORICS} 2019 - 24th European Symposium on
  Research in Computer Security, Luxembourg, September 23-27, 2019,
  Proceedings, Part {II}}, 2019, pp. 484--503.

\bibitem{BertinoCertLessDrone}
J.~{Won}, S.~{Seo}, and E.~{Bertino}, ``Certificateless cryptographic protocols
  for efficient drone-based smart city applications,'' \emph{IEEE Access},
  vol.~5, pp. 3721--3749, 2017.

\bibitem{CertLessYang}
G.~Yang and C.-H. Tan, ``Strongly secure certificateless key exchange without
  pairing,'' in \emph{Proceedings of the 6th ACM Symposium on Information,
  Computer and Communications Security}, ser. ASIACCS.\hskip 1em plus 0.5em
  minus 0.4em\relax New York, NY, USA: ACM, 2011, pp. 71--79.

\end{thebibliography}
